\documentclass{article}

\usepackage[preprint, nonatbib]{neurips_2023}

\usepackage[utf8]{inputenc} % allow utf-8 input
\usepackage[T1]{fontenc}    % use 8-bit T1 fonts
\usepackage{hyperref}       % hyperlinks
\usepackage{url}            % simple URL typesetting
\usepackage{booktabs}       % professional-quality tables
\usepackage{amsfonts}       % blackboard math symbols
\usepackage{nicefrac}       % compact symbols for 1/2, etc.
\usepackage{microtype}      % microtypography
\usepackage{xcolor}         % colors
\usepackage[numbers]{natbib}

\usepackage{enumerate}
\usepackage{algorithm}
\usepackage[noend]{algpseudocode}

\newcommand{\algorithmicbreak}{\textbf{break}}
\newcommand{\Break}{\algorithmicbreak}

\usepackage{amsmath}
\usepackage{amssymb}
\usepackage{mathtools}
\usepackage{amsthm}
\usepackage{dsfont}
\usepackage[capitalize,noabbrev]{cleveref}

\usepackage[noend]{algpseudocode}
\usepackage{tikz,pgfplots}

\theoremstyle{plain}
\newtheorem{theorem}{Theorem}[section]

\newtheorem{lemma}[theorem]{Lemma}
\newtheorem{corollary}[theorem]{Corollary}
\theoremstyle{definition}
\newtheorem{definition}[theorem]{Definition}

\theoremstyle{remark}

\newenvironment{proofof}[1]{\begin{trivlist} \item {\bf Proof
#1:~~}}
  {\qed\end{trivlist}}
\renewenvironment{proofof}[1]{\par\medskip\noindent{\bf Proof of #1: \ }}{\hfill$\Box$\par\medskip}

\usepackage[textsize=tiny]{todonotes}

\usepackage{tcolorbox}

\usepackage{framed}
\usepackage{mdframed}

\newcommand{\REAL}{\ensuremath{\mathbb{R}}}
\newcommand{\NATURAL}{\ensuremath{\mathbb{N}}}
\newcommand{\eps}{\varepsilon}

\newcommand{\blue}[1]{#1}

\newcommand{\Ex}[1]{\ensuremath{\mathbf{E}[#1]}}

\DeclareMathOperator{\argmax}{arg\,max}

\newcommand{\subsetselect}{{\textsc{SubsetSelection}}}

\newcommand{\uniformsubset}{{\textsc{UniformSubset}}}
\newcommand{\localsearchsubset}{{\textsc{LocalSearchSubset}}}
\newcommand{\dynamicthreshold}{{\textsc{DynamicThreshold }}} 
\newcommand{\dynamicmonotone}{{\textsc{DynamicThresholding}}}

\newcommand{\InsertF}{{\textsc{Insert}}}
\newcommand{\DeleteF}{{\textsc{Delete}}}
\newcommand{\UpdateF}{{\textsc{Update}}}
\newcommand{\extractF}{{\textsc{Extract}}}

\newcommand{\constLevel}{\textsc{ConstructLevel }}
\newcommand{\insertv}{{\textsc{Insert}}}
\newcommand{\deletev}{{\textsc{Delete}}}

\newcommand{\mO}{O}
\newcommand{\mI}{{\mathcal{I}}}

\newcommand{\etal}{\emph{et al.}}

\title{Dynamic Non-monotone Submodular Maximization}

\author{
Kiarash Banihashem\thanks{equal contribution} \textsuperscript{ \thanks{Department of Computer Science, University of Maryland, College Park, MD, USA.}}
\\ \texttt{ kiarash@umd.edu} \\ University of Maryland%, College Park 
\And
Leyla Biabani\footnotemark[1] \textsuperscript{\thanks{Department of Mathematics and Computer Science, Eindhoven University of Technology, the Netherlands.} }
\\ \texttt{ l.biabani@tue.nl} \\ TU Eindhoven 
\And
Samira Goudarzi\footnotemark[1]\textsuperscript{\ \ \footnotemark[2]} \\ \texttt{ samirag@umd.edu} \\ University of Maryland%, College Park 
\And
MohammadTaghi Hajiaghayi\footnotemark[1]\textsuperscript{\ \ \footnotemark[2]}\\  \texttt{ hajiagha@cs.umd.edu} \\ University of Maryland%, College Park 
\And
Peyman Jabbarzade\footnotemark[1]\textsuperscript{\ \ \footnotemark[2]} \\ \texttt{ peymanj@umd.edu} \\ University of Maryland%, College Park 
\And
Morteza Monemizadeh\footnotemark[1]\textsuperscript{\ \ \footnotemark[3]} \\ \texttt{ m.monemizadeh@tue.nl} \\ TU Eindhoven 
}

\begin{document}

\maketitle

\begin{abstract}
Maximizing submodular functions has been increasingly used in many applications of machine learning, such as data summarization, recommendation systems, and feature selection. Moreover, there has been a growing interest in both submodular maximization and dynamic algorithms. 
In 2020, Monemizadeh~\cite{DBLP:conf/nips/Monemizadeh20} and 
Lattanzi, Mitrovic, Norouzi{-}Fard, Tarnawski, and Zadimoghaddam~\cite{DBLP:conf/nips/LattanziMNTZ20} initiated developing dynamic algorithms for the monotone submodular maximization problem under the cardinality constraint $k$. 
In 2022, Chen and Peng ~\cite{DBLP:journals/corr/abs-2111-03198} studied the complexity of this problem and raised an important open question: "\emph{Can we extend [fully dynamic] results (algorithm or hardness) to non-monotone submodular maximization?}". 
We affirmatively answer their question by demonstrating a reduction from maximizing a non-monotone submodular function under the cardinality constraint $k$ to maximizing a monotone submodular function under the same constraint. 
Through this reduction, we obtain the first dynamic algorithms solving the non-monotone submodular maximization problem under the cardinality constraint $k$. 
We've derived two algorithms, both maintaining an $(8+\epsilon)$-approximate of the solution. 
The first algorithm requires $\mathcal{O}(\epsilon^{-3}k^3\log^3(n)\log(k))$ oracle queries per update, while the second one requires $\mathcal{O}(\epsilon^{-1}k^2\log^3(k))$. 
Furthermore, we showcase the benefits of our dynamic algorithm for video summarization and max-cut problems on several real-world data sets.
\end{abstract}

%-----------------------------------------------------------------------------------------
\section{Introduction}
Submodular functions are powerful tools for solving
real-world problems as they provide a theoretical framework for modeling the famous 
``\emph{diminishing returns}''~\cite{DBLP:journals/mp/Fujishige84a} phenomenon arising in a variety of practical settings.
Many theoretical problems such as those involving graph cuts, entropy-based
clustering, coverage functions,
 and mutual information can be cast in the submodular maximization framework.
As a result, submodular functions have been increasingly used
in many applications of machine learning such as data summarization~\cite{NIPS2014_a8e864d0,DBLP:conf/cikm/SiposSSJ12,DBLP:conf/iccv/SimonSS07}, feature selection~\cite{DBLP:conf/stoc/DasK08,DBLP:journals/jmlr/DasK18,DBLP:conf/icml/DasK11,DBLP:conf/aistats/KhannaEDNG17}, 
and recommendation systems~\cite{DBLP:conf/kdd/El-AriniG11}. 
These applications include both the monotone and non-monotone versions of the maximization of submodular functions.

%-----------------------------------------------------------------------------------------
\paragraph{Applications of non-monotone submodular maximization.}
The general problem of non-monotone submodular maximization
has been studied extensively in~\cite{DBLP:journals/siamcomp/FeigeMV11, 10.5555/2634074.2634180, DBLP:journals/siamcomp/BuchbinderFNS15,pmlr-v48-mirzasoleiman16, 10.5555/3327144.3327162, DBLP:conf/icml/Norouzi-FardTMZ18}.
% Non-monotone submodular maximization
This problem has numerous  applications in
video summarization, movie recommendation~\cite{pmlr-v48-mirzasoleiman16}, and revenue maximization in viral marketing~\cite{hartline2008optimal}\footnote{The
problem of selecting a subset of people in a social network to maximize their influence in a viral marketing campaign 
can be modeled as a constrained submodular maximization problem. When we introduce a cost, 
then the influence minus the cost is modeled as non-monotone submodular maximization problems~\cite{DBLP:conf/approx/BateniHZ10,DBLP:journals/talg/BateniHZ13,DBLP:conf/wine/GuptaRST10}.
}. An important application of
% non-monotone submodular maximization 
this problem
appears in maximizing the difference between a monotone submodular function and a linear function 
that penalizes the
addition of more elements to the set (e.g., the coverage
and diversity trade-off). An illustrative example of this application is the maximum facility location  
in which we want to open a subset of facilities and maximize the total profit from served clients plus the
cost of facilities we did not open~\cite{DBLP:conf/iccv/DueckF07}.
Another important application occurs when expressing learning problems such as feature selection
using weakly submodular functions~\cite{DBLP:conf/stoc/DasK08, DBLP:conf/aistats/KhannaEDNG17, DBLP:journals/corr/ElenbergKDN16, 10.5555/3454287.3454743}. 

\paragraph{Our contribution.}
In this paper, we consider the non-monotone submodular maximization problem 
under cardinality constraint $k$ 
in the \emph{fully dynamic setting}.
In this model, we have a \emph{universal ground set} $V$.
At any time $t$, ground set $V_t \subseteq V$ is the set of elements that are inserted but not deleted after their last insertion till time $t$.
More formally, we assume that there is a sequence of ``updates'' such that each update either inserts 
an element to $V_{t-1}$ or deletes an element from $V_{t-1}$ to form $V_t$. 

We assume that there is a (non-monotone) submodular function $f$ 
that is defined over the universal ground set $V$.
Our goal is to maintain, at each point in time, a set of size at most $k$ whose submodular value 
is maximum among any subset of $V_t$ of size at most $k$.

Since calculating such a set is known to be NP-hard ~\cite{DBLP:journals/siamcomp/FeigeMV11} 
even in the offline setting (where you get all the items at the same time), 
we focus on providing algorithms with provable approximation guarantees, 
while maintaining fast update time. This is challenging as elements may be inserted or deleted, 
possibly in an adversarial order. While several dynamic algorithms exist 
for monotone submodular maximization, 
non-monotone submodular maximization is a considerably more challenging problem as adding elements to a set may decrease its value.

In STOC 2022, Chen and Peng~\cite{DBLP:journals/corr/abs-2111-03198} raised the following open question:

%\begin{tcolorbox}[width=\linewidth, colback=white!90!gray,boxrule=0pt]
\textbf{Open problem:}  
 ``Can we extend [fully dynamic] results (algorithm or hardness) to non-monotone submodular maximization?''
%\end{tcolorbox}

 In this paper, we answer their question affirmatively by providing
 the first dynamic algorithms for non-monotone submodular maximization.

To emphasize the significance of our result, it should be considered that although monotone submodular maximization under cardinality constraint has a tight $\frac{e}{e-1}$ approximation algorithm in the offline mode and nearly tight ($2 + e$) approximation algorithms for both streaming and dynamic settings, there is a hardness result for the non-monotone version stating that it is impossible to obtain a $2.04$ (i.e., $0.491$) approximation algorithm for this problem even in the offline setting\cite{10.5555/2133036.2133119}, and to the best of our knowledge, the current state of the art algorithms for this problem have $2.6$ and $3.6$ (i.e., $0.385$ and $0.2779$) approximation guarantees for offline \cite{DBLP:journals/mor/BuchbinderF19} and streaming settings \cite{DBLP:conf/icalp/AlalufEFNS20}, respectively.
 
We obtain our result, by proposing a general reduction from the problem of dynamically maintaining 
a non-monotone submodular function under cardinality constraint $k$ to developing a dynamic thresholding algorithm for maximizing monotone submodular functions under the same constraint.
We first define $\tau$-thresholding dynamic algorithms that we use in our reduction.

%\begin{tcolorbox}[width=\linewidth, colback=white!95!gray]
%\begin{tcolorbox}[width=\linewidth, colback=white!90!gray,boxrule=0pt]
\begin{definition}[$\tau$-Thresholding Dynamic Algorithm]
\label{def:tau:threshold}
Let $\tau > 0$ be a parameter. 
We say a dynamic algorithm is $\tau$-thresholding if at any time $t$ of sequence $\Xi$, it 
reports a set $S_t \subseteq V_t$ of size at most $k$ such that 
\begin{itemize}
    \item \textbf{Property 1:} either\textbf{ a)}  $S_t$ has $k$ elements and $f(S_t) \geq k\tau$, or \textbf{b)} $S_t$ has less than $k$ elements and for any $v\in V_t\setminus S_t$, the marginal gain $\Delta(v|S_t) < \tau$.
    \item \textbf{Property 2:}  The number of elements changed in any update, i.e, 
    $|S_{t+1} \backslash S_{t}| + |S_{t} \backslash S_{t+1}|$,
    is not more than the number of queries made by the algorithm during the update.
\end{itemize}

\end{definition}
%\end{tcolorbox}

%-----------------------------------------------------------------------------------------
 
In the above definition, the first property reflects the main intuition of threshold-based algorithms, while the last property is a technical condition required in our analysis. 
It's worth noting that the thresholding technique has been used widely for optimizing submodular functions \cite{DBLP:conf/nips/Monemizadeh20,DBLP:conf/soda/LiuV19,DBLP:conf/icml/FahrbachMZ19,DBLP:conf/nips/LattanziMNTZ20,chen2022practical}.
We next state our main result, which is a general reduction.

%-----------------------------------------------------------------------------------------

%\begin{tcolorbox}[width=\linewidth, colback=white!90!gray,boxrule=0pt]
\begin{theorem}[\textbf{Reduction Metatheorem}]
\label{thm:meta}
Suppose that $f:2^V \rightarrow \REAL_{\geq 0}$  is a (possibly non-monotone) submodular function 
defined on subsets of a ground set $V$ and let $k \in \NATURAL$ be a parameter. 

Assume that for any given value of $\tau$, there exists a $\tau$-thresholding dynamic algorithm with an expected (amortized) $ \mO (g(n,k))$ oracle queries per update.
Then, there exist the following dynamic algorithms: 
\begin{itemize}
  \item A dynamic algorithm with an approximation guarantee of $(8+\eps)$ 
 using an expected (amortized) $\mO(k+\min(k, g(n,k)) \cdot g(n,k) {\cdot \eps^{-1}\log(k)} )$ oracle queries per update.
    \item A dynamic algorithm maintaining a $(10+\eps)$-approximate solution of the optimal value of $f$  using an expected (amortized)
 $\mO(\min(k, g(n,k)) \cdot g(n,k) {\cdot \eps^{-1}\log(k)} )$ oracle calls per update.
\end{itemize}
\end{theorem}
%\end{tcolorbox}
%-----------------------------------------------------------------------------------------
In ~\cite{DBLP:conf/nips/Monemizadeh20}, Monemizadeh developed a dynamic algorithm for monotone submodular maximization under cardinality constraint $k$, which requires an amortized $\mO(\eps^{-2}k^2\log^3(n))$ number of oracle queries per update.
Interestingly, in the appendix, we show that this algorithm is indeed $\tau$-thresholding (for any given $\tau$).  
Now, if we use this $\tau$-thresholding dynamic algorithm 
inside our reduction Metatheorem ~\ref{thm:meta}, we obtain a dynamic algorithm that maintains a $(8+\eps)$-approximate solution using an expected amortized $\mO(\eps^{-3}k^3\log^3(n)\log(k))$ 
oracle queries per update.

The recent paper \cite{bani2023dynamicmat} of Banihashem, Biabani, Goudarzi, Hajiaghayi, Jabbarzade, and Monemizadeh develops a new dynamic algorithm for monotone submodular maximization under cardinality constraint $k$, which uses an expected $\mO(\eps^{-1}k\log^2(k))$ number of oracle queries per update. A similar proof shows that this new algorithm is $\tau$-thresholding as well. We have provided its pseudocode and a detailed explanation on why this algorithm is indeed $\tau$-thresholding in the appendix. By exploiting this algorithm in our Reduction Metatheorem ~\ref{thm:meta}, we can reduce the number of oracle queries mentioned to an expected number of $\mO(\eps^{-2}k^2\log^3(k))$ per update.

The second result in Theorem~\ref{thm:meta} is also of interest as it can 
be used to devise a dynamic algorithm for non-monotone submodular maximization with polylogarithmic query complexity if one can provide a $\tau$-thresholding dynamic algorithm 
for maximizing monotone submodular functions (under the cardinality constraint $k$) with polylogarithmic query complexity.

%-----------------------------------------------------------------------------------------

\subsection{Preliminaries}
\paragraph{Submodular maximization.}
Let $V$ be a ground set of elements. 
We say a function $f:2^V\to \REAL_{\geq 0}$  is a \textit{submodular} function 
if for any $A, B \subseteq V$, $f(A)+f(B) \geq f(A\cup B) + f(A\cap B)$.
Equivalently, $f$ is a submodular function if for any subsets $A \subseteq B \subseteq V$ and
for any element $e \in V\setminus B$, it holds that 
$f(A\cup \{e\}) - f(A) \geq f(B \cup \{e\}) - f(B) \enspace .$
We define $\Delta(e|A) := f(A \cup \{e\}) - f(A)$ the \emph{marginal gain} of adding the element $e$ to set $A$.
Similarly, we define $\Delta(B|A) := f(A \cup B) - f(A)$ for any sets $A, B \subseteq V$.
Function $f$ is \textit{monotone} if $f(A) \leq f(B)$ holds for any $A \subseteq B \subseteq V$, and it is \textit{non-monotone} if it is not necessarily the case. 
In the submodular maximization problem under cardinality constraint $k$, we seek to compute 
a set $S^*$ such that $|S^*| \leq k$ and $f(S^*)=\max_{|S|\leq k}f(S)$, 
where $f$ is a submodular function and $k \in \mathbb{N}$ is a given parameter.

\paragraph{Query access model.}
Similar to recent dynamic works~\cite{DBLP:conf/nips/LattanziMNTZ20, DBLP:journals/corr/abs-2111-03198}, 
we assume the access to a submodular function $f$ is given by an \emph{oracle}. 
The oracle allows \emph{set queries} such that 
for every subset $A \subseteq V$, one can query the value $f(A)$. 
In this query access model, the marginal gain $\Delta_f(e | A) \doteq f(A \cup \{e\}) - f(A)$
for every subset  $A \subseteq V$ and an element $e \in V\backslash A$, 
 can be computed using two set queries. 
To do so, we first query $f(A \cup \{e\})$ and then $f(A)$.

%-----------------------------------------------------------------------------------------

\paragraph{Dynamic model.}
Let $\Xi$ be a sequence of inserts and deletes of an underlying universe $V$. 
We assume that $f: 2^V \rightarrow \mathbb{R}_{\ge 0}$ is a (possibly non-monotone) submodular function 
defined on subsets of the universe $V$. 
We define time $t$ to be the $t\textsuperscript{th}$ update (i.e., insertion or deletion) of sequence $\Xi$. 
We let $\Xi_t$ be the sub-sequence of updates from the beginning of sequence $\Xi$ till time $t$ 
and denote by $V_t \subseteq V$ the set of elements that are inserted but not deleted from the beginning of the sequence $\Xi$ 
till any time $t$. That is, $V_t$ is the current ground set of elements. 
We let  $OPT_t=\max_{S \subseteq V_t: |S|\le k}  f(S)$.

\paragraph{Query complexity.} 
The \emph{query complexity} of a dynamic $\alpha$-approximate algorithm is the number of oracle queries that the algorithm must make to compute a solution $S_t$ with respect to ground set $V_t$ 
whose submodular value is an $\alpha$-approximation of $OPT_t$. 
That is, $|S_t| \le k$ and $f(S_t) \ge \alpha\cdot OPT_t$. 
Observe that the dynamic algorithm remembers every query it has made so far. Thus results of queries
made in previous times may help find $S_t$ in current time $t$.

%-----------------------------------------------------------------------------------------

\paragraph{Oblivious adversarial model.}
The dynamic algorithms that we develop in this paper are in the \emph{oblivious adversarial model} as is common for analysis
of randomized data structures such as universal hashing \cite{DBLP:conf/stoc/CarterW77}. 
The model allows the adversary, who is aware of the submodular function $f$ and the algorithm that is going to be used, to determine all the arrivals and departures of the elements in the ground set $V$.
However, the adversary is unaware of the random bits used in the algorithm and so cannot choose updates adaptively in response to the randomly guided choices of the algorithm. Equivalently, we can suppose that the adversary prepares the full input (insertions and deletions) before the algorithm runs.

%-----------------------------------------------------------------------------------------
%-----------------------------------------------------------------------------------------

\subsection{Related Work}
\label{sec.related}
\paragraph{Offline algorithms.} 
The offline version of non-monotone submodular maximization  
was first studied by Feige, Mirrokni, and Vondr{\'{a}}k in~\cite{DBLP:journals/siamcomp/FeigeMV11}. 
They studied \emph{unconstrained non-monotone submodular maximization} and developed constant-factor approximation algorithms for this problem. 
% Indeed, 
In the offline query access model, they showed that  a subset
$S$ chosen uniformly at random has a submodular value which is a $4$-approximation of 
the optimal value for this problem. 
In addition,  they also described two local search algorithms. 
The first uses $f$ as the objective function, and provides $3$-approximation and 
the second uses a noisy version of $f$ as the objective function and achieves an improved
approximation guarantee $2.5$ for maximizing unconstrained non-monotone non-negative submodular functions. 
Interestingly, they showed $(2-\eps)$-approximation for symmetric submodular functions would require 
an exponential number of queries for any fixed $\eps > 0$. 

Oveis Gharan and Vondr{\'{a}}k~\cite{DBLP:conf/soda/GharanV11}
showed that an extension of the $2.5$-approximation algorithm can be seen as \emph{simulated annealing} method which provides an improved approximation of roughly $2.4$. 
Later, Buchbinder, Feldman, Naor, and Schwartz~\cite{DBLP:journals/siamcomp/BuchbinderFNS15} 
at FOCS'12, presented a  randomized linear time algorithm achieving a tight approximation guarantee of $2$ 
that matches the known hardness result of~\cite{DBLP:journals/siamcomp/FeigeMV11}. 
Bateni, Hajiaghayi, and Zadimoghaddam~\cite{DBLP:conf/approx/BateniHZ10,DBLP:journals/talg/BateniHZ13} and 
Gupta, Roth, Schoenebeck, and Talwar~\cite{DBLP:conf/wine/GuptaRST10} independently studied non-monotone submodular maximization 
subject to cardinality constraint $k$ in the offline and secretary settings. 
In particular, Gupta~\etal~\cite{DBLP:conf/wine/GuptaRST10} 
obtained an offline $6.5$-approximation for this problem.

%-----------------------------------------------------------------------------------------

All of the aforementioned approximation algorithms are offline, where the whole input is given 
in the beginning, whereas the need for real-time analysis of rapidly changing data streams 
motivates the study of this problem in settings such as the dynamic model that we study in this paper.
% \lb{Note that the secretary model is online (irrevocable)}
 
%-----------------------------------------------------------------------------------------

\paragraph{Streaming algorithms.} 
The dynamic model that we study in this paper is closely related to 
the streaming model~\cite{DBLP:journals/jcss/AlonMS99,DBLP:journals/jacm/Indyk06}. 
However, the difference between these two models is 
that in the streaming model, we maintain a data structure using which 
we compute a solution at the end of the stream and so, the time to extract the solution 
is not important as we do it once. However, in the dynamic model, 
we need to maintain a solution after every update, thus, the update time 
of a dynamic algorithm should be as fast as possible. 

The known streaming algorithms~\cite{DBLP:conf/aaai/MirzasoleimanJ018,DBLP:conf/nips/FeldmanK018,DBLP:conf/icalp/FeldmanNS11} 
    work in the insertion-only streaming model and they do not support deletions as well as insertions.  
    Indeed, there are streaming algorithms~\cite{DBLP:conf/icml/0001ZK18,DBLP:conf/icml/MirzasoleimanK017} 
    for the monotone submodular maximization problem that support deletions, but the space and the update time 
    of these algorithms depend on the number of deletions which could be $\Omega(n)$, 
   where $n= |V|$ is the size of ground set $V$.

%-----------------------------------------------------------------------------------------

For monotone submodular maximization, 
Badanidiyuru, Mirzasoleiman, Karbasi, and Krause \cite{Badanidiyuru2014StreamingSM}
proposed an insertion-only streaming algorithm with a
$(2+ \eps)$-approximation guarantee under a cardinality constraint $k$.
Chekuri, Gupta, and Quanrud~\cite{10.1007/978-3-662-47672-7_26} presented (insertion-only) streaming algorithms for maximizing monotone 
and non-monotone submodular functions subject to $p$-matchoid constraint\footnote{
For non-monotone submodular maximization subject to cardinality constraint $k$, 
Chekuri, Gupta, and Quanrud~\cite{10.1007/978-3-662-47672-7_26} claimed that they obtained 
$4.7$-approximation algorithm. However, 
Alaluf, Ene, Feldman, Nguyen, and Suh~\cite{DBLP:conf/icalp/AlalufEFNS20} 
found an error in the proof of this approximation guarantee.}. 
Later, Mirzasoleiman, Jegelka, and Krause \cite{DBLP:conf/aaai/MirzasoleimanJ018} and 
Feldman, Karbasi, and Kazemi~\cite{DBLP:conf/nips/FeldmanK018} developed 
streaming algorithms with better approximation guarantees 
for maximizing a non-monotone function under a $p$-matchoid constraint. 
Currently, the best streaming algorithm for maximizing a non-monotone submodular function subject to a cardinality
constraint is due to Alaluf, Ene, Feldman, Nguyen, and Suh~\cite{DBLP:conf/icalp/AlalufEFNS20} 
whose approximation guarantee is $3.6+ \eps$, improving the $5.8$-approximation guarantee 
that was proposed by 
Feldman~\etal~\cite{DBLP:conf/nips/FeldmanK018}.

%-----------------------------------------------------------------------------------------

\paragraph{Dynamic algorithms.}
At NeurIPS 2020,
Lattanzi, Mitrovic, Norouzi{-}Fard, Tarnawski, and Zadimoghaddam~\cite{DBLP:conf/nips/LattanziMNTZ20} 

and Monemizadeh \cite{DBLP:conf/nips/Monemizadeh20} initiated the study of submodular maximization in the dynamic model.
They presented dynamic algorithms that maintain $(2+\eps)$-approximate solutions for 
 maximizing a monotone submodular function subject to cardinality constraint $k$.
Later, at STOC 2022, Chen and Peng~\cite{DBLP:journals/corr/abs-2111-03198} studied the complexity of this problem 
and they proved that developing a $c$-approximation dynamic algorithm for $c < 2$ is not possible 
unless we use a number of oracle queries polynomial in the size of ground set $V$.  
In 2023, Banihashem, Biabani, Goudarzi, Hajiaghayi, Jabbarzade, and Monemizadeh\cite{pmlr-v202-banihashem23a} developed an algorithm for monotone submodular maximization problem under cardinality constraint $k$ using a polylogarithmic amortized update time. 
Concurrent works of Banihashem, Biabani, Goudarzi, Hajiaghayi, Jabbarzade, and Monemizadeh\cite{bani2023dynamicmat} and 
Duetting, Fusco, Lattanzi, Norouzi-Fard, and Zadimoghaddam\cite{pmlr-v202-duetting23a} developed the first dynamic algorithms for monotone submodular maximization under a matroid constraint. Authors of \cite{bani2023dynamicmat} also improve the algorithm of \cite{DBLP:conf/nips/Monemizadeh20} for monotone submodular maximization subject to cardinality constraint $k$. 
There are also studies on the dynamic model of influence maximization, which shares similarities with submodular maximization~\cite{DBLP:conf/nips/Peng21}.

In this paper, for the first time, we study the generalized version of their problem by presenting an algorithm for maximizing the non-monotone submodular functions in the dynamic setting. 

%-----------------------------------------------------------------------------------------
\section{Dynamic algorithm}
\label{sec:dyn:alg}
%-----------------------------------------------------------------------------------------

In this section, we explain the algorithm that we use in the reduction that we stated in Metatheorm~\ref{thm:meta}. 
The pseudocode of our algorithm is  given in Algorithm~\ref{alg:init}, Algorithm~\ref{alg:update}, 
and Algorithm~\ref{alg:subset}.

Such reductions were previously proposed in the
offline model by \cite{DBLP:conf/wine/GuptaRST10}, and
later works extended this idea
to the streaming model \cite{10.1007/978-3-662-47672-7_26, DBLP:conf/aaai/MirzasoleimanJ018}.
We develop a reduction in the dynamic model inspired by these works, 
though in our proof,
we require a tighter analysis to obtain the approximation guarantee in our setting.

We consider an arbitrary time $t$ of sequence $\Xi$ where $V_t$ is 
the set of elements inserted before time $t$, but not deleted after their last insertion.  
Let $OPT_t^*=\max_{S \subseteq V_t: |S|\le k}  f(S)$. 
For simplicity, we drop $t$ from $V_t$ and $OPT_t^*$, when it is clear from the context.

{In the following, we assume that the value of $OPT$ is known. 
Although the exact value of $OPT^*$ is unknown, we can maintain parallel runs of 
our dynamic algorithm for different guesses of the optimal value.
By using $(1+\eps')^i$, where $i\in \mathbb{Z}$ as our guesses for the optimal value, one of our guesses $(1+\eps')$-approximates the value of $OPT^*$. We show that the output of our algorithm satisfies the approximation guarantee in the run whose $OPT$ $(1+\eps')$-approximates the value of $OPT^*$.
Later, in the appendix, we show that it is enough to consider each element $e$ only in runs $i$ 
for which we have 
$\frac{\eps'}{k}\cdot{(1+\eps')^i} \leq f(e) \leq (1+\eps')^i$.
This method increases the query complexity of our dynamic algorithm by only a factor of $\mO(\eps^{-1}\log{k})$.}

Our approach for solving the non-monotone submodular maximization is to first run the thresholding algorithm with input set $V$ 
to find a set $S_1$ of at most $k$ elements. 
Since $f$ is non-monotone, subsets of $S_1$ might have a higher submodular value than $f(S_1)$.
Then, we use an $\alpha$-approximation algorithm (for $0 < \alpha \leq 1$) to choose
a set $S'_1 \subseteq S_1$ with guarantee $\Ex{f(S'_1)} \ge \alpha\cdot \max_{C \subseteq S_1} f(C)$. 
Next, we run the thresholding algorithm with the input set 
$V \backslash S_1$ and compute a set $S_2$. 
At the end, we return set $S = \arg\max_{C \in \{S_1,S'_1,S_2\}} f(C)$. 
Intuitively, for an optimal solution $S^*$, if $f(S_1 \cap S^*)$ is a good approximation of $OPT$, 
then $f(S'_1)$ is a good approximation of $OPT$. On the other hand, if both $f(S_1)$ and $f(S_1 \cap S^*)$ are small 
with respect to $OPT$, 
then we can ignore the elements of $S_1$ and show that we can find a set $S_2 \subseteq V\setminus S_1$ 
of size at most $k$ whose submodular value is a good approximation of $OPT$. 
The following lemma proves that the submodular value of $S$ is a reliable approximation of the optimal solution. 
The formal proof of this lemma can be found in Section~\ref{subsec:analysis}.

\begin{lemma}[Approximation Guarantee]
\label{lem:approxi}
Assuming that $OPT^* \in [\frac{OPT}{1 + \eps'}, OPT]$, the expected submodular value of set $S$ 
is $\Ex{f(S)} \ge (1-\mO(\eps'))\frac{OPT^*}{6+\frac{1}{\alpha}}$. 
\end{lemma}

Next, we explain the steps of our reduction in detail. 

Let us first fix the threshold $\tau=\frac{OPT}{k(3+1/(2\alpha))}$. Then, we fix a $\tau$-thresholding dynamic algorithm (for example, \cite{DBLP:conf/nips/Monemizadeh20} or \cite{bani2023dynamicmat}) and suppose we denote it by $\dynamicmonotone$. 
Before sequence $\Xi$ of updates starts, we create two independent instances $\mathcal{I}_1$ and $\mathcal{I}_2$ of $\dynamicmonotone$. 
The first instance will maintain set $S_1$ 
and the second instance will maintain set $S_2$. 
For instance $\mathcal{I}_i$ where $i \in \{1,2\}$, we consider the following subroutines:
\begin{itemize}
  \item $\InsertF_{\mathcal{I}_i}(v)$: This subroutine inserts an element $v$ to  instance $\mathcal{I}_i$.
    \item $\DeleteF_{\mathcal{I}_i}(v)$:
    Invoking this subroutine will delete the element $v$ from instance $\mathcal{I}_i$.
    \item $\extractF_{\mathcal{I}_i}$: This subroutine returns the maintained set (of size at most $k$) of 
    $\mathcal{I}_i$.
\end{itemize}

%-----------------------------------------------------------------------------------------

\paragraph{Extracting $S_1$.} 
After the update at time $t$,
first, we would like to set $Z=S^-_{1} \cup \{v\}$ or $Z=S^-_{1}\setminus \{v\}$, if the update is the insertion of an element $v$ or the deletion of an element $v$, respectively, where $S^-_{1}$ is the set $S_1$ that instance $\mathcal{I}_1$ maintains just before this update.
To find set $S_1^-$, we just need to invoke subroutine $\extractF_{\mathcal{I}_1}$.
If the update is an insertion, we insert it into instance $\mathcal{I}_1$ using 
$\InsertF_{\mathcal{I}_1}(v,\tau)$, and if the update is a deletion, we delete $v$ 
{from both $\mathcal{I}_1$ and $\mathcal{I}_2$ using 
$\DeleteF_{\mathcal{I}_1}(v)$ and $\DeleteF_{\mathcal{I}_2}(v)$.} 
We then invoke $\extractF_{\mathcal{I}_1}$ once again to return set $S_1$.

\paragraph{Extracting $S'_1$.} 
Buchbinder~\etal~\cite{DBLP:journals/siamcomp/BuchbinderFNS15}
developed  a method to extract a subset $S'_1 \subseteq S_1$ whose submodular value is 
a good approximation of $\max_{C \subseteq S_1} f(C)$.
In this algorithm, we start with two solutions $\emptyset$ and $S_1$. 
The algorithm considers the elements (in arbitrary order) one at a time. 
For each element, it determines whether it should be added
to the first solution or removed from the second solution.
Thus, after a single pass over set  $S_1$, both
solutions completely coincide, which is the solution that the
algorithm outputs. They show that a (deterministic) greedy choice in each
step obtains $3$-approximation of the best solution in $S_1$. 
However, if we combine this greedy choice with randomization, 
we can obtain a $2$-approximate solution.
Since we do a single pass over set $S_1$, the number of oracle queries is $\mO(|S_1|)$.

The second algorithm that we can use to extract $S'_1$ is a random sampling algorithm proposed by 
Feige~\etal~\cite{DBLP:journals/siamcomp/FeigeMV11},
which choose every element in $S_1$ with probability $1/2$.
They show that this random sampling returns a set $S'_1$ whose approximation factor is $1/4$ of $\max_{C \subseteq S_1} f(C)$, 
and its number of oracle calls is $\mO(1)$.  
We denote either of these two methods by \subsetselect.

\paragraph{Extracting $S_2$.}
Next, we would like to update the set $S_2$ that is maintained by instance $\mathcal{I}_2$.  
To do this, for every element 
$u \in Z\backslash S_1$, we add it to  $\mathcal{I}_2$ using 
$\InsertF_{\mathcal{I}_2}(u,\tau)$, and for every element 
$u \in S_1\backslash Z$, 
we delete it from  $\mathcal{I}_2$ using $\DeleteF_{\mathcal{I}_2}(u,\tau)$. 
Finally, when $\mathcal{I}_2$ exactly includes all the current elements other than the ones in $S_1$,  we call subroutine $\extractF_{\mathcal{I}_2}$ to return set $S_2$.

\begin{corollary}
\label{corol:approx}
    We obtain the $(8+\eps)$ approximation guaranty stated in the Metatheorem ~\ref{thm:meta} by using the local search method for our $\subsetselect$, and we get the $(10+\eps)$ approximation guaranty by utilizing the random sampling method for our $\subsetselect$ subroutine.
\end{corollary}
\begin{proof}
    These are immediate results of Lemma~\ref{lem:approxi}, and $\alpha$ being $\frac{1}{2}$ and $\frac{1}{4}$ in the local search method and random sampling method, respectively.
\end{proof}

%-----------------------------------------------------------------------------------------

\begin{algorithm}[h]
  \caption{Initialization$(k, OPT)$}
  \label{alg:init}
  \begin{algorithmic}[1]
        \State $\tau \gets \frac{OPT}{k(3+1/(2\alpha))}$, where $\alpha$ is $\frac{1}{2}$ or $\frac{1}{4}$ based on the selection of algorithm for $\subsetselect$.
        \State Instantiate two independent instances $\mathcal{I}_1$ and $\mathcal{I}_2$  of 
      $\dynamicmonotone$ for monotone submodular maximization under cardinality constraint $k$ using $\tau$
    \end{algorithmic}
\end{algorithm}

\begin{algorithm}[h]
  \caption{\textsc{Update}$(v)$}
  \label{alg:update}
  \begin{algorithmic}[1]
    \State $Z \gets \extractF_{\mathcal{I}_1}$  
      %\State $G_1^- \gets M_1.G$ \Comment{Old value of $G_1$}
    \If{\textsc{Update}$(v)$ is an insertion}
        \State Invoke $\InsertF_{\mathcal{I}_1}(v)$, \
        {$Z \gets Z \cup \{v\}$}
    \Else
        \State Invoke $\DeleteF_{\mathcal{I}_1}(v)$, \ $\DeleteF_{\mathcal{I}_2}(v)$, \
        {$Z \gets Z \setminus \{v\}$}
    \EndIf
    \State $S_{1} \gets \extractF_{\mathcal{I}_1}$
    \State $S_1' \gets \subsetselect{}(S_1)$
    \For{$u \in S_1\backslash Z$}
        \State $\DeleteF_{\mathcal{I}_2}(u)$
    \EndFor
    \For{$u \in Z\backslash S_1$}
        \State $\InsertF_{\mathcal{I}_2}(u)$
    \EndFor
    \State $S_{2} \gets \extractF_{\mathcal{I}_2}$
    \State Return $ \argmax_{C \in \{S_1, S_1', S_2\}} f(C)$
  \end{algorithmic}
\end{algorithm}

%-----------------------------------------------------------------------------------------

\begin{algorithm}[h]
  \caption{\subsetselect$(S)$}
  \label{alg:subset}
  \begin{algorithmic}[1]
    \Function{\uniformsubset}{$S$}
      \State $T\gets \emptyset$
      \For{$s\in S$}
        \If{$Coin(\frac{1}{2}$)} \Comment{With probability $\frac{1}{2}$}
          \State $T \gets T \cup \{s\}$
        \EndIf
      \EndFor
      \State \Return $T$
    \EndFunction
    \Function{\localsearchsubset}{$S$} 
        \State $X_0 \gets \emptyset$, \ $Y_0 \gets S$.
        \For{$i=1$ to $|S|$}
            \State  $a_i  \gets f(X_{i-1}\cup \{s_i\})-f(X_{i-1})$, \  $b_i  \gets f(Y_{i-1}\setminus \{s_i\})-f(Y_{i-1})$
            % \State  $b_i  \gets f(Y_{i-1}\setminus \{s_i\})-f(Y_{i-1})$
            \State $a'_i \gets \max(a_i,0)$, \ $b'_i \gets \max(b_i,0)$
            \If{$a'_i=b'_i=0$} 
            	$a'_i/(a'_i+b'_i)=0$
	    \EndIf
            \State \textbf{with probability} $a'_i/(a'_i+b'_i)$ \textbf{do:}
            \State $X_i\gets X_{i-1} \cup \{s_i\}$, \ $Y_i \gets Y_{i-1}$
            \State  \textbf{else do:} $X_i\gets X_{i-1}$, \ $Y_i \gets Y_{i-1}\setminus \{s_i\}$
        \EndFor
        \State \Return  $X_{|S|}$ (or equivalently $Y_{|S|}$)
    \EndFunction
  \end{algorithmic}
\end{algorithm}

%-----------------------------------------------------------------------------------------

%-----------------------------------------------------------------------------------------

\paragraph{Analysis.}
\label{subsec:analysis}
In this section, we prove the correctness of our algorithms and analyze the number of oracle queries 
of our algorithms, which finishes the proof of Theorems~\ref{thm:meta}. 

Consider an arbitrary time $t$. Let $S_t$ be the reported set of $\dynamicmonotone$ at time $t$. 
Recall that $V_t$ is the ground set at time $t$, and we drop the $t$ for simplicity, so we use $V$ and $S$ to denote $V_t$ and $S_t$. 
We first present Lemma~\ref{lm:anupam_threshold} whose proof is given in the appendix. 
Then we proceed to prove Lemma~\ref{lem:approxi} and Theorem \ref{thm:meta}

\begin{lemma}
\label{lm:anupam_threshold}
Suppose that set $S$ satisfies Property 1.b of Definition \ref{def:tau:threshold}. 
It means that $S$ has less than $k$ elements and for any $v\in V\setminus S$, the marginal gain $\Delta(v|S) < \tau$.
Then, for any arbitrary subset $C \subseteq V$, we have $f(S) \ge f(S \cup C) - |C| \cdot \tau$.
\end{lemma}
\begin{proofof}{Lemma~\ref{lem:approxi}}
Assume that at a fixed time $t$, $OPT^*$ and $S^*$ are the optimal value and an optimal solution for the submodular maximization of function $f$ under cardinality constraint $k$. This means that $|S^*| \le k$ and $f(S^*) = OPT^*$.
Recall that $\tau=\frac{OPT}{k(3+\frac{1}{2\alpha})}$, where $OPT$ is our guess for the optimal value.
Also, by assumption we have $OPT^* \in [\frac{OPT}{1 + \epsilon'}, OPT]$, or equivalently $OPT \in [OPT^*, (1 + \epsilon')OPT^*]$.

To prove the lemma, we claim that 
$\max(\Ex{f(S_1)}, \Ex{f(S_1')}, \Ex{f(S_2)}) \geq (1-\mO(\eps'))\frac{OPT^*}{6 + \frac{1}{\alpha}}$. 

Suppose this claim is true. 
Using Jensen's inequality~\cite{DBLP:books/cu/10/D2010}, we have
$\Ex{\max(f(S_1), f(S_1'), f(S_2))} \geq \max(\Ex{f(S_1)}, \Ex{f(S_1')}, \Ex{f(S_2)})$ 
, which yields
$\Ex{f(S)} \geq (1-\mO(\eps'))\frac{OPT^*}{6+1/\alpha}$. 

To prove the claim,
we consider two cases. The first case is when $f(S_1 \cap S^*) \geq \frac{\tau k}{2\alpha}$ 
and the second case is if $f(S_1 \cap S^*) < \frac{\tau k}{2\alpha}$.

Suppose the first case is true. 
Then, the subset selection algorithm (either random sampling method or local search) 
returns $S_1'$ for which $\Ex{f(S_1')} \geq \alpha\cdot\max_{S\subseteq S_1}f(S)$.
Since $S_1 \cap S^* \subseteq S_1$, we have 

For the latter case, we show that $\Ex{f(S_1)}+ \Ex{f(S_2)} \geq (1 - \mO(\eps'))\frac{OPT^*}{3+1/2\alpha}$, inferring $\max{(\Ex{f(S_1)}, \Ex{f(S_2)})} \geq (1 - \mO(\eps'))\frac{OPT^*}{6+1/\alpha}$. 
Indeed, since $S_1$ and $S_2$ are reported by an $\tau$-thresholding algorithm, 
if $|S_1| = k$ or $|S_2| = k$, then $\max(\Ex{f(S_1)}, \Ex{f(S_2)})$ is at least $\tau k = \frac{OPT}{3+1/2\alpha}$
by the first property of $\tau$-thresholding algorithms.

Now suppose that $|S_1|, |S_2| < k$, which means that Property 1.b of Definition \ref{def:tau:threshold} holds for $S_1$ and $S_2$.
Therefore, we have $f(S_1) \geq f(S_1\cup S^*)-\tau|S^*|$ and $f(S_2) \geq f(S_2 \cup (S^*\setminus S_1))-\tau|S^*\setminus S_1|$ by Lemma \ref{lm:anupam_threshold}. Besides, we have $f(S_1 \cap S^*) - \frac{\tau k}{2\alpha} < 0$. Therefore,
\[
        f(S_1)+ f(S_2) \geq f(S_1\cup S^*)-\tau|S^*| + f(S_2 \cup (S^*\setminus S_1))\\ 
                     -\tau|S^*\setminus S_1| + f(S_1 \cap S^*) - \frac{\tau k}{2\alpha} \enspace .
\]
Since $|S^*\setminus S_1| \leq |S^*| \leq k$ we have
\[
        f(S_1)+ f(S_2) \geq f(S_1\cup S^*) + f(S_2 \cup (S^*\setminus S_1)) \\
+ f(S_1 \cap S^*) - (2+1/(2\alpha))\tau k \enspace .
\]

Since $S_1\cap S_2 = \emptyset$ and $f$ is submodular, we have 
$f(S_1 \cup S^*) + f (S_2 \cup (S^*\setminus S_1)) \geq
f(S_1 \cup S_2 \cup S^*) + f(S^*\setminus S_1)$.
Additionally, by the submodularity and non-negativity of $f$, we have
$f(S_1 \cap S^*) \geq f(S^*) - f (S^*\setminus S_1)$, because $f (S^*\setminus S_1) + f(S_1 \cap S^*) \geq f(S^*) + f(\emptyset)$. By adding the last two inequalities and using the non-negativity of $f$ once again, we get
$f(S_1 \cup S^*) + f (S_2 \cup (S^*\setminus S_1))+ f(S_1 \cap S^*) \geq f(S_1 \cup S_2 \cup S^*) + f(S^*) \geq f(S^*) = OPT^*$.
By putting everything together we have,
\[
    f(S_1)+ f(S_2) \geq OPT^* - (2+1/(2\alpha))\tau k = OPT^* - (\frac{4\alpha + 1}{2\alpha})(\frac{OPT(2\alpha)}{6\alpha + 1}).
\]
By using the assumption that $OPT \leq (1 + \epsilon')OPT^*$, we have, 
\[
    f(S_1)+ f(S_2) \geq OPT^* (1 - (\frac{(4\alpha + 1)(1 + \epsilon')}{6\alpha + 1})) \geq OPT^* (\frac{2\alpha-\epsilon'(4\alpha + 1)}{6\alpha + 1}) = (1-\mO(\eps'))\frac{OPT^* }{3 + 1/(2\alpha)}.
\]
\end{proofof}

\begin{proofof}{Theorem \ref{thm:meta}}

As previously discussed, we've established in Lemma~\ref{lem:approxi} and Corollary \ref{corol:approx} that utilizing the local search method for the \subsetselect{} subroutine results in an approximation ratio of $(8+\eps)$, whereas the random sampling method achieves an approximation ratio of $(10+\eps)$. Thus, the only remaining aspect to address in proving this theorem is proving the query complexity of our proposed algorithm. In Lemma~\ref{lem:time}, we bound the number of queries made in each run of our algorithm per update, proving the bounds given in Theorem~\ref{thm:meta} by considering the extra $\mO(\eps^{-1}\log{k})$-factor caused by our parallel runs.
\end{proofof}
\begin{lemma}
\label{lem:time}
     Let the random variable $Q_t$ denote the number of oracle calls that our algorithm \blue{in Theorem~\ref{thm:meta}} makes at time $t$ in each of the parallel runs. Depending on whether the expected or expected amortized number of oracle calls made by the thresholding algorithm \dynamicthreshold  per each update is $\mO(g(n,k))$, 
    if we choose the local search method as our \subsetselect{} subroutine, we have
    % \begin{align}
    \[
        \Ex{Q_t} \in \mO (\min(k \cdot g(n, k), g(n, k)^2))
        + \mO(k) \enspace,
    \]
    or 
    \[
        \Ex{\sum_{t=1}^T Q_t} \in \mO (T \cdot \min(k \cdot g(n, k), g(n, k)^2))
        + \mO(k) \enspace,
    \]
    
    and if we choose the random sampling method as our \subsetselect{} subroutine, we have
    \[
        \Ex{Q_t}  \in \mO (\min(k \cdot g(n, k), g(n, k)^2)) \enspace ,
    \]
    or 
    \[
        \Ex{\sum_{t=1}^T Q_t} \in \mO (T \cdot \min(k \cdot g(n, k), g(n, k)^2)) \enspace
    \]
\end{lemma}
\begin{proof}
Consider the case where the expected number of oracle calls made by the thresholding algorithm \dynamicthreshold per each update is $\mO(g(n,k))$. Per each update, our algorithm makes an update in instance $\mI_1$ causing $\mO(g(n,k))$ oracle queries. Next, we make either $\mO(k)$ or $0$ oracle queries for the \subsetselect{} subroutine, depending on the used method. We also make a series of updates in instance $\mI_2$, each causing $\mO(g(n,k))$ oracle queries. The number of such updates is bounded by the number of changes in the output of instance $\mI_1$, which is bounded by both $k$ and $\mO(g(n,k))$ (according to the second property of Definition \ref{def:tau:threshold}). These comprise all the oracle queries made by our algorithm at time $t$. Therefore, the given bounds for this case hold. A detailed proof for the remaining bounds is provided in the appendix. 
\end{proof}

\section{Empirical results}
\label{section:empirical_results}
In this section, we empirically study our $(8+\eps)$-approximation dynamic algorithm.
We implement our codes in C++ and run them on a MacBook laptop with $8$ GB RAM and $M1$ processor.
We empirically study the performance of our algorithm for video summarization and the Max-Cut problem.  

\paragraph{Video summarization.}
Here, we use the Determinantal Point Process (DPP) which is introduced by \cite{macchi_1975}, and combine it with our algorithm to capture a video summarization. We run our experiments on YouTube and Open Video Project (OVP) datasets from \cite{DBLP:journals/prl/AvilaLLA11}.

For each video, we use the linear method of \cite{DBLP:conf/nips/GongCGS14} to extract a subset of frames and find a positive semi-definite kernel $L$ with size $n \times n$ where $n$ is the number of extracted frames. Then, we try to find a subset $S$ of frames such that it maximizes $\frac{det(L_S)}{det(I+L)}$ where $L_S$ is the sub-matrix of $L$ restricted to indices corresponding to frames $S$.
Since $L$ is a positive semi-definite matrix, we have $det(L_S) \ge 0$.
Interestingly, \cite{DBLP:journals/ftml/KuleszaT12} showed that $\log(det(L_S))$ is a non-monotone function.
We use these properties and set $f(S):= \log(det(L_S) + 1)$ to make $f$ a non-monotone non-negative submodular function. 
Then we run our $(8+\eps)$-approximate dynamic algorithm to find the best $S$ to maximize $f(S)$ such that $|S|\le k$ for $k \in [10]$.

\begin{figure}[h]
\begin{center}
\includegraphics[scale=0.2]{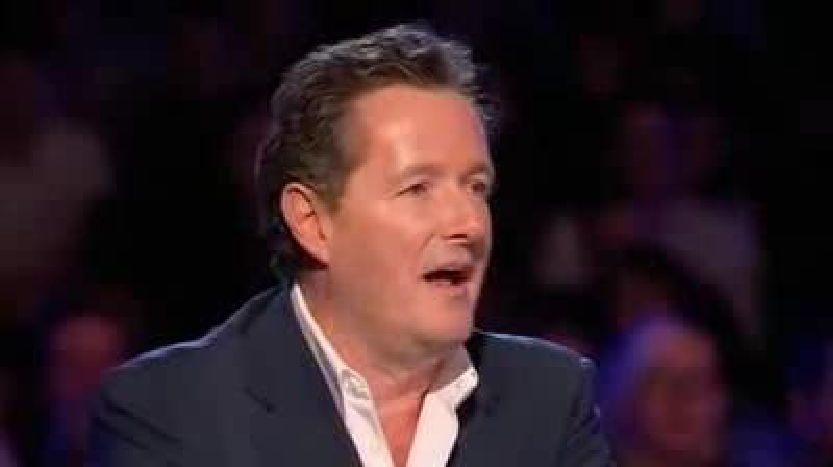}
\includegraphics[scale=0.2]{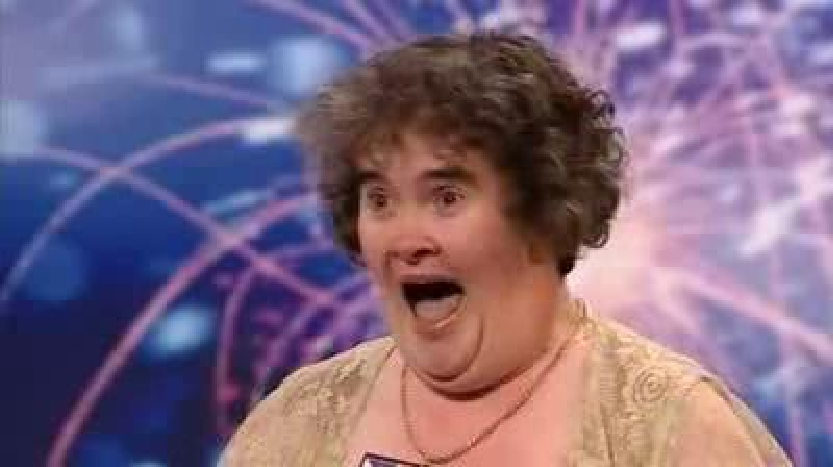}
\includegraphics[scale=0.2]{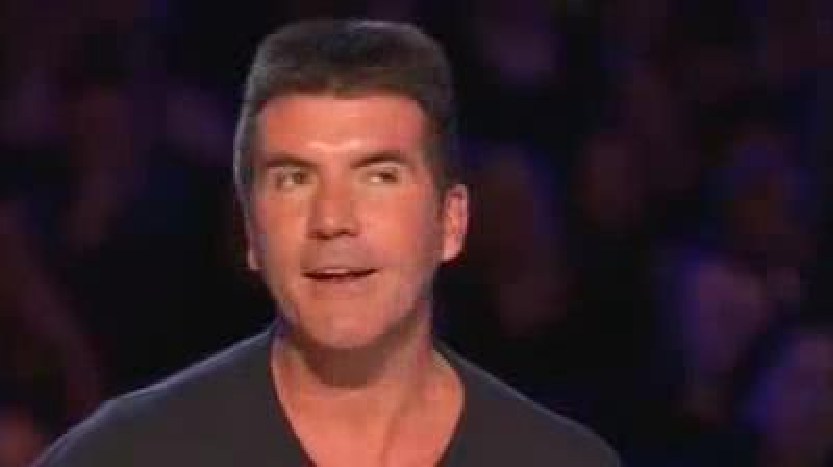}
\includegraphics[scale=0.2]{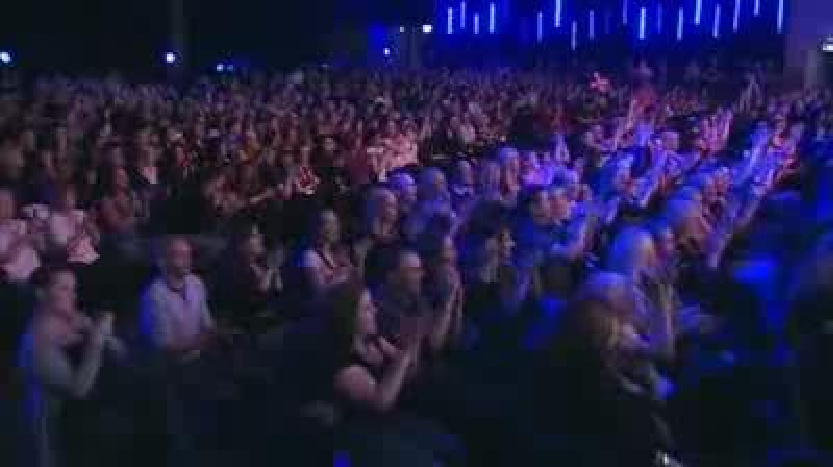}
\end{center}
\vspace{-0.1cm}
\caption{
Video summarization of Susan Boyle's performance on Britain's Got Talent show (video 106) from YouTube.}
\label{fig:youtube_images}
\end{figure}

\begin{figure}[h]
\begin{center}
\includegraphics[scale=0.2]{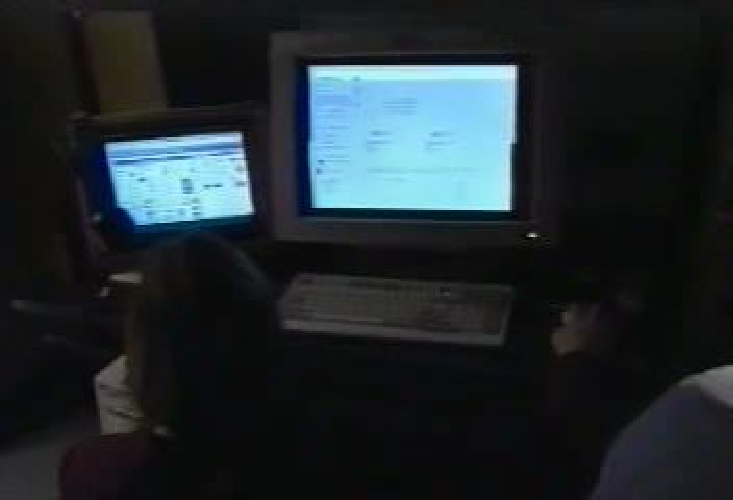}
\includegraphics[scale=0.2]{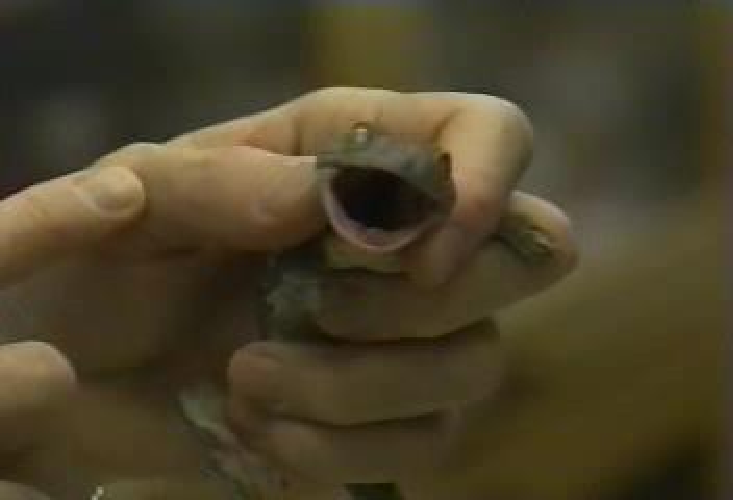}
\includegraphics[scale=0.2]{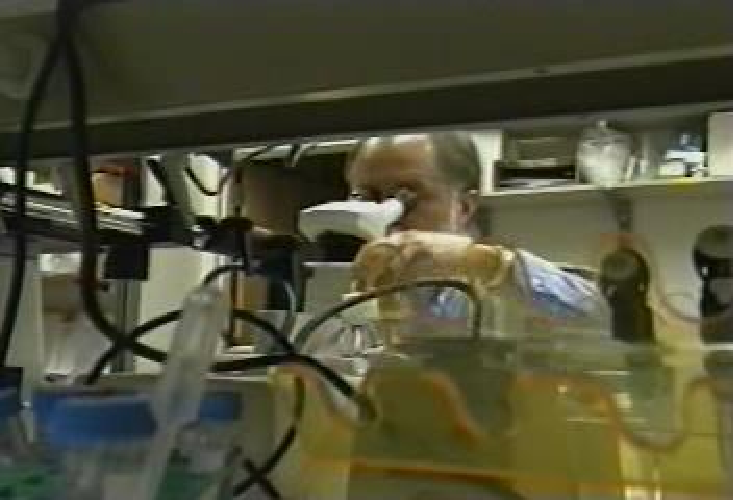}
\includegraphics[scale=0.2]{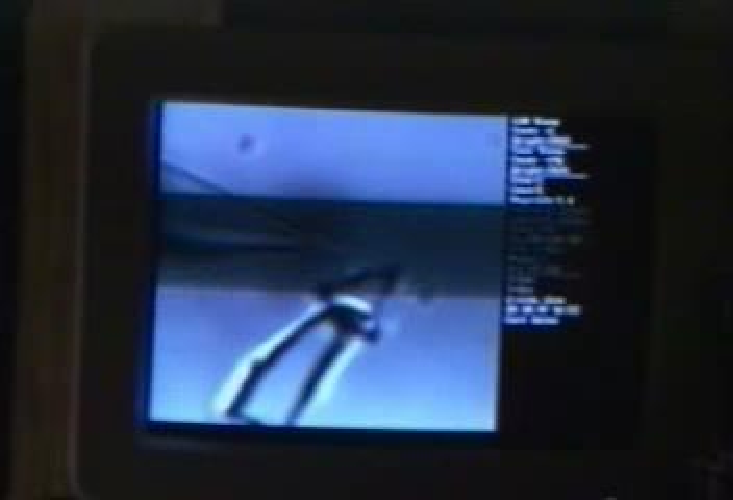}
\end{center}
\vspace{-0.1cm}
\caption{Video summarization for "Senses And Sensitivity, Introduct. to Lecture 1 presenter" (video 36) from OVP.}
\label{fig:OVP_images}
\end{figure}

First, we insert all frames to observe the quality of our algorithm. Figure \ref{fig:youtube_images} and \ref{fig:OVP_images} 
are the selected frames by our algorithm for Video 106 from YouTube and Video 36 from OVP, respectively, 
when we limit the number of selected frames to $4$. 
Then, we create a sequence $\Xi$ of updates of 
frames of each video. Similar to \cite{DBLP:conf/nips/LattanziMNTZ20}, 
we define the sequence as a sliding window model. 
That is, given a window of size $W$ for a parameter $W \in \mathbb{N}$, 
a frame is inserted at a time $t$ and will be alive for a window of size $W$ 
and then we delete that frame. 

To evaluate the performance of our algorithm, 
we benchmark (See Figure~\ref{fig:video_experiment_benchmark}) the total number of query calls and the submodular value of set $S$ 
of our algorithm and the streaming algorithm proposed for non-monotone submodular maximization so-called \textsc{Sample-Streaming} 
proposed in~\cite{DBLP:conf/nips/FeldmanK018}. 
This algorithm works as follows: 
Upon arrival of an element $u$, with probability $(1-q)$, for a parameter $0 < q < 1$, 
we ignore $u$, otherwise (i.e., with probability $q$), we do the following. 
If the size of set $S$ that we maintain is less than  $k$, i.e, 
$|S|<k$ and $\Delta(u|S) > 0$, we add $u$ to $S$.  
However, if $|S| = k$, we select an element $v \in S$ for which $\Delta(v:S)$ is minimum possible, 
where $\Delta(u:S)$ equals to $\Delta(u | S_u)$ where $S_u$ are elements that arrived before $u$ 
in sequence $\Xi$. If $\Delta(u|S) \ge (1+c) \Delta(v:S)$ for a constant $c$, we replace $v$ by $u$; 
otherwise, we do nothing. 
Now we convert this streaming algorithm into a dynamic algorithm. 
To accomplish this, we restart \textsc{Sample-Streaming} after every deletion that deletes an element of 
solution set $S$ that is reported by \textsc{Sample-Streaming}'s outputs. That is, if a deletion does not 
touch any element in set $S$, we do nothing; otherwise we restart the streaming algorithm.

%!TEX root = main.tex
% \vspace{-1.2em}
\begin{figure*}[h]
\centering
%\subfigure[loc-Gowalla's total oracle calls]{
\begin{tabular}{@{}c@{}c@{}c@{}}
\begin{tikzpicture}[scale=0.76]
\pgfplotsset{width=4.8cm,compat=1.9}
\begin{axis}[
    xlabel={$k$},
    ylabel={oracle calls},
    ymin=0, ymax=1100,
     xtick={2, 4, 6, 8, 10},
    legend pos=north west,
    ymajorgrids=true,
    grid style=dashed,
    legend style={nodes={scale=0.5, transform shape}},
]

\addplot[
    color=red,
    mark=square,
    error bars/.cd,
    y dir=both,
    y explicit,
    error bar style={
    color=red}
    ]
    coordinates {
        (1, 331.20000) += (0, 106.80000) -= (0, 123.20000) (2, 373.60000) += (0, 108.40000) -= (0, 143.60000) (3, 346.80000) += (0, 142.20000) -= (0, 64.80000) (4, 371.00000) += (0, 77.00000) -= (0, 124.00000) (5, 329.80000) += (0, 94.20000) -= (0, 89.80000) (6, 309.50000) += (0, 50.50000) -= (0, 50.50000) (7, 368.80000) += (0, 80.20000) -= (0, 127.80000) (8, 352.70000) += (0, 54.30000) -= (0, 74.70000) (9, 367.60000) += (0, 112.40000) -= (0, 104.60000) (10, 385.60000) += (0, 169.40000) -= (0, 94.60000) 
    };
    \addlegendentry{\textsc{Sample-Streaming}}

\addplot[
    color=blue,
    mark=star,
    error bars/.cd,
    y dir=both,
    y explicit,
    error bar style={
    color=blue}
    ]
    coordinates {
        (1, 461.40000) += (0, 9.60000) -= (0, 15.40000) (2, 662.30000) += (0, 90.70000) -= (0, 99.30000) (3, 529.50000) += (0, 91.50000) -= (0, 64.50000) (4, 447.80000) += (0, 31.20000) -= (0, 54.80000) (5, 323.90000) += (0, 20.10000) -= (0, 24.90000) (6, 235.40000) += (0, 16.60000) -= (0, 12.40000) (7, 179.20000) += (0, 9.80000) -= (0, 16.20000) (8, 178.50000) += (0, 18.50000) -= (0, 10.50000) (9, 189.80000) += (0, 16.20000) -= (0, 16.80000) (10, 238.30000) += (0, 14.70000) -= (0, 18.30000)
    };
    \addlegendentry{\textsc{Our Dynamic Algorithm}}
    
\end{axis}
\end{tikzpicture}
% \caption{loc-Gowalla's total oracle calls}
% % \end{tabular}
%

% \small (c) Video 106 total oracle calls

% 
%\subfigure[loc-Gowalla's average output]{
% \begin{tabular}{@{}c@{}}

% \end{tabular}
% %

% \small (a) Video 106 total oracle calls and average output

% % 
% %\subfigure[loc-Gowalla's average output]{
% \begin{tabular}{@{}c@{}c@{}c@{}}

\begin{tikzpicture}[scale=0.76]
\pgfplotsset{width=4.8cm,compat=1.9}
\begin{axis}[
    xlabel={$k$},
    ylabel={oracle calls},
    ymin=0, ymax=1300,
     xtick={2, 4, 6, 8, 10},
    legend pos=north west,
    ymajorgrids=true,
    grid style=dashed,
    legend style={nodes={scale=0.5, transform shape}},
]

\addplot[
    color=red,
    mark=square,
    error bars/.cd,
    y dir=both,
    y explicit,
    error bar style={
    color=red}
    ]
    coordinates {
        (1, 322.70000) += (0, 146.30000) -= (0, 146.70000) (2, 484.00000) += (0, 183.00000) -= (0, 309.00000) (3, 553.50000) += (0, 217.50000) -= (0, 183.50000) (4, 687.60000) += (0, 171.40000) -= (0, 182.60000) (5, 603.70000) += (0, 118.30000) -= (0, 142.70000) (6, 625.30000) += (0, 94.70000) -= (0, 112.30000) (7, 681.60000) += (0, 220.40000) -= (0, 104.60000) (8, 615.80000) += (0, 133.20000) -= (0, 196.80000) (9, 717.10000) += (0, 173.90000) -= (0, 180.10000) (10, 694.10000) += (0, 170.90000) -= (0, 180.10000) 
    };
    \addlegendentry{\textsc{Sample-Streaming}}

\addplot[
    color=blue,
    mark=star,
    error bars/.cd,
    y dir=both,
    y explicit,
    error bar style={
    color=blue}
    ]
    coordinates {
        (1, 445.80000) += (0, 14.20000) -= (0, 15.80000) (2, 654.50000) += (0, 86.50000) -= (0, 56.50000) (3, 580.50000) += (0, 41.50000) -= (0, 51.50000) (4, 574.30000) += (0, 60.70000) -= (0, 51.30000) (5, 675.40000) += (0, 156.60000) -= (0, 88.40000) (6, 644.90000) += (0, 89.10000) -= (0, 79.90000) (7, 581.80000) += (0, 52.20000) -= (0, 39.80000) (8, 520.90000) += (0, 21.10000) -= (0, 47.90000) (9, 583.10000) += (0, 53.90000) -= (0, 55.10000) (10, 706.50000) += (0, 173.50000) -= (0, 107.50000) 
    };
    \addlegendentry{\textsc{Our Dynamic Algorithm}}
    
\end{axis}
\end{tikzpicture}
% % \caption{loc-Gowalla's total oracle calls}
% \end{tabular}
% %

\begin{tikzpicture}[scale=0.76]
\pgfplotsset{width=4.8cm,compat=1.9}
\begin{axis}[
    xlabel={$k$},
    ylabel={f},
    ymin=0, ymax=6,
    xtick={2, 4, 6, 8, 10},
    legend pos=north west,
    ymajorgrids=true,
    grid style=dashed,
    legend style={nodes={scale=0.5, transform shape}},
]

\addplot[
    color=red,
    mark=square,
    error bars/.cd,
    y dir=both,
    y explicit,
    error bar style={
    color=red}
    ]
    coordinates {
        (1, 1.62818) += (0, 0.44989) -= (0, 0.57743) (2, 2.14678) += (0, 0.65564) -= (0, 0.41891) (3, 2.48939) += (0, 0.33780) -= (0, 0.40127) (4, 2.39747) += (0, 0.34600) -= (0, 0.64868) (5, 2.28432) += (0, 0.52396) -= (0, 0.47845) (6, 2.30636) += (0, 0.45690) -= (0, 0.35722) (7, 2.43273) += (0, 0.48661) -= (0, 0.59661) (8, 2.39775) += (0, 0.46000) -= (0, 0.46214) (9, 2.28745) += (0, 0.42656) -= (0, 0.32351) (10, 2.39145) += (0, 0.50267) -= (0, 0.74555) 
    };
    \addlegendentry{\textsc{Sample-Streaming}}

\addplot[
    color=blue,
    mark=star,
    error bars/.cd,
    y dir=both,
    y explicit,
    error bar style={
    color=blue}
    ]
    coordinates {
        (1, 2.46816) += (0, 0.07933) -= (0, 0.42347) (2, 2.15307) += (0, 0.32738) -= (0, 0.17542) (3, 2.29888) += (0, 0.09777) -= (0, 0.19832) (4, 2.19944) += (0, 0.43743) -= (0, 0.58491) (5, 2.39385) += (0, 0.40503) -= (0, 0.38268) (6, 2.36480) += (0, 0.17151) -= (0, 0.35363) (7, 2.47486) += (0, 0.06145) -= (0, 0.17877) (8, 2.49553) += (0, 0.03520) -= (0, 0.02626) (9, 2.46536) += (0, 0.07095) -= (0, 0.25866) (10, 2.26592) += (0, 0.20335) -= (0, 0.26033)
    };
    \addlegendentry{\textsc{Our Dynamic Algorithm}}

% \addlegendentry{\textsc{Random Algorithm}}
    
\end{axis}
\end{tikzpicture}

% \small (a) Video 36 total oracle calls

% % 
% %\subfigure[loc-Gowalla's average output]{
% \begin{tabular}{@{}c@{}}
\begin{tikzpicture}[scale=0.76]
\pgfplotsset{width=4.8cm,compat=1.9}
\begin{axis}[
    xlabel={$k$},
    ylabel={f},
    ymin=0, ymax=10,
    xtick={2, 4, 6, 8, 10},
    legend pos=north west,
    ymajorgrids=true,
    grid style=dashed,
    legend style={nodes={scale=0.5, transform shape}},
]

\addplot[
    color=red,
    mark=square,
    error bars/.cd,
    y dir=both,
    y explicit,
    error bar style={
    color=red}
    ]
    coordinates {
        (1, 2.20679) += (0, 0.38838) -= (0, 0.38206) (2, 3.30840) += (0, 0.44028) -= (0, 0.39134) (3, 3.51886) += (0, 0.73748) -= (0, 0.68497) (4, 3.22474) += (0, 0.36300) -= (0, 0.61065) (5, 2.83012) += (0, 0.86126) -= (0, 1.18679) (6, 2.91402) += (0, 0.50340) -= (0, 1.78629) (7, 2.96153) += (0, 0.91239) -= (0, 1.41191) (8, 3.19881) += (0, 0.46492) -= (0, 0.93366) (9, 3.18183) += (0, 0.96342) -= (0, 1.56121) (10, 3.23328) += (0, 0.66669) -= (0, 0.46605)
    };
    \addlegendentry{\textsc{Sample-Streaming}}

\addplot[
    color=blue,
    mark=star,
    error bars/.cd,
    y dir=both,
    y explicit,
    error bar style={
    color=blue}
    ]
    coordinates {
        (1, 3.34082) += (0, 0.00204) -= (0, 0.00613) (2, 3.62122) += (0, 0.31347) -= (0, 0.27428) (3, 4.40612) += (0, 0.52449) -= (0, 0.53673) (4, 3.27633) += (0, 0.29102) -= (0, 0.32123) (5, 3.82531) += (0, 0.28081) -= (0, 0.55184) (6, 4.18082) += (0, 0.52530) -= (0, 0.44204) (7, 4.47143) += (0, 0.29184) -= (0, 0.36531) (8, 4.73959) += (0, 0.41143) -= (0, 0.69061) (9, 4.56939) += (0, 0.92449) -= (0, 0.63878) (10, 4.52245) += (0, 0.91020) -= (0, 0.66531)
    };
    \addlegendentry{\textsc{Our Dynamic Algorithm}}

% \addlegendentry{\textsc{Random Algorithm}}
    
\end{axis}
\end{tikzpicture}
% \caption*{loc-Gowalla's average output}
% \end{tabular}
\\
\small \hspace{11mm} (a) \hspace{32mm} (b) \hspace{27mm} (c) \hspace{29mm} (d)
% \small (b) Video 36 average output

% \begin{tabular}{@{}c@{}}

% \caption*{loc-Gowalla's average output}
\end{tabular}

% \small (b) Video 36 total oracle calls and average output

\caption{We plot the total number of query calls and the average output of our dynamic algorithm and \textsc{Sample-Streaming} 
on video 106 from YouTube and video 36 from OVP. 
In this figure, from left to right, 
Sub-figures (a) and (b) are the total oracle calls for video 106 and 36, respectively.  
Similarly, Sub-figures (c) and (d) are average submodular value for video 106 and 36, respectively.  
}
\label{fig:video_experiment_benchmark}
\end{figure*}
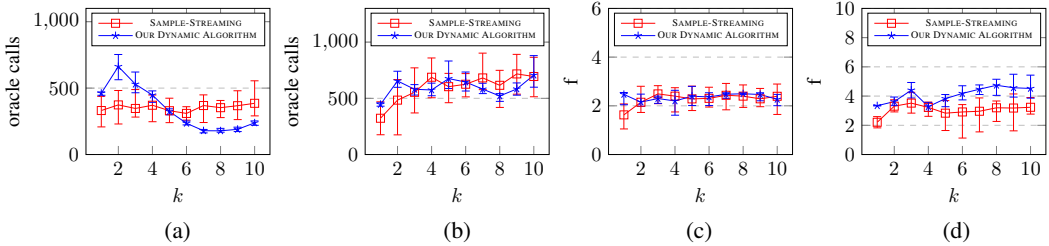

We run our algorithm for $\eps = k/2$ and compare the total oracle calls and average output of our algorithm and \textsc{Sample-Streaming} in Figure \ref{fig:video_experiment_benchmark}.
To prove the approximation guarantee of our dynamic algorithm, we assumed $\eps \le 1$.
However, in practice, it is possible to increase $\eps$ up to a certain level without affecting the output of the algorithm significantly. 
On the other hand, increasing $\eps$ reduces the total oracle calls and makes the algorithm faster.
As you can see in Figure \ref{fig:video_experiment_benchmark} plots (b) and (d), the submodular value of our algorithm is not worse than the \textsc{Sample-Streaming} algorithm whose approximation factor is $3+2\sqrt{2} \approx 5.828$ which is better than our approximation factor.
Thus, our algorithm has an outcome better than our expectation, while its total oracle calls are better than \textsc{Sample-Streaming} algorithm (look at Figure \ref{fig:video_experiment_benchmark} plots (a) and (c)). 

We also empirically study the celebrated Max-Cut problem which is a non-monotone submodular maximization function (See~\cite{DBLP:journals/siamcomp/FeigeMV11}). 
These experiments are given in the appendix.

\section{Conclusion}
In this paper, we studied non-monotone submodular maximization subject to cardinality constraint $k$ in the dynamic setting by providing a reduction from this problem to maximizing monotone submodular functions under the cardinality constraint $k$ with a certain kind of algorithms($\tau$-thresholding algorithms). 
Moreover, we used our reduction to develop the first dynamic algorithms for this problem. 
In particular, both our algorithms maintain a solution set whose submodular value is a $(8+\eps)$-approximation of the optimal value and require $\mO(\eps^{-3}k^3\log^3(n)\log(k))$ and $\mO(\eps^{-1}k^2\log^3(k))$ oracle queries per update, respectively. 

\section{Acknowledgements}

This work is partially supported by DARPA QuICC NSF AF:Small \#2218678, and NSF AF:Small \#2114269.

\bibliographystyle{plain}
\bibliography{references}

\newpage 
\appendix

\section{Guessing optimal value}
The dynamic algorithm that we proposed in the main part of the paper, 
assumes the value of $OPT$ is given as a parameter.
However, we know that the value of $OPT^*$ is not known; it changes after every insertion or deletion, and our final goal in this algorithm is to find a correct approximation of it. 
As discussed earlier, to achieve this goal, we run parallel instances of 
our dynamic algorithm for different guesses of the optimal value. We use  $(1+\eps')^i$, where $i\in \mathbb{Z}$ as our guesses for the optimal value. We showed that at any fixed time, the run in our algorithm whose corresponding $OPT$ $(1 + \eps')$-approximates $OPT^*$ also finds a solution with a guaranteed approximation factor of $(8+\eps)$ or $(10 + \eps)$, depending on the value of $\alpha$. 
However, If $\rho$ is the ratio between the maximum and minimum non-zero possible value of the optimal solution at any time $t$, then the number of parallel instances of our algorithm will be 
$\mO(\log_{(1+\eps')}\rho) = \mO(\eps^{-1}\cdot \log\rho)$. 
This incurs an extra $\mO(\eps^{-1}\cdot \log\rho)$-factor in the query complexity of our dynamic algorithm. 

Next, we show how to replace this extra factor with 
an extra factor of $\mO(\eps^{-1}\cdot \log k)$, which is independent of $\rho$. 
We use the well-known technique that has been also used in \cite{DBLP:conf/nips/LattanziMNTZ20}. 
In particular, for every element $v$, we only add it to those instances $i$ 
for which we have $\frac{\eps'}{k}\cdot{(1+\eps')^i} \leq f(v) \leq (1+\eps')^i = OPT$.
Note that for run $i$ corresponding to the guess of optimal value $OPT = (1+\eps')^i$,
$f(v) > (1+\eps')^i$ means that $f(v)$ is greater than $OPT$. Therefore, we can safely ignore this element for this run. In fact, this run might be useful for finding the solution only if $OPT$ is a $(1+\eps')$ approximation of $OPT^*$, which can only happen after the deletion of $v$, so ignoring $v$ has no effect on the solution and saves query complexity. 
 
Moreover, ignoring all elements $v$ whose $f(v) < \frac{\eps'}{k}\cdot{(1+\eps')^i}$ in the run $i$ corresponding to the guess of optimal value $OPT = (1+\eps')^i$, decreases the submodular value of its solution by at most $OPT\cdot \eps'$. This changes the approximation guarantee by $\mO(\eps')$, which can be ignored.

In this way, every element $v$ is added to at most  $\mO(\eps^{-1}\log{k})$ parallel instances. 
Thus, after every insertion or deletion, 
we need to update only $\mO(\eps^{-1}\log{k})$ instances of our dynamic algorithm.
\section{Proof of Lemma \ref{lm:anupam_threshold}}

\newtheorem*{recall:lm:anupam_threshold}{Lemma~\ref{lm:anupam_threshold}}
\begin{recall:lm:anupam_threshold}
Suppose set $S$ satisfies Property 2 of Definition \ref{def:tau:threshold}. 
In other words, $S$ has less than $k$ elements and for any $v\in V\setminus S$, the marginal gain $\Delta(v|S) < \tau$.
Then, for any arbitrary subset $C \subseteq V$, we have $f(S) \ge f(S \cup C) - |C| \cdot \tau$.
\end{recall:lm:anupam_threshold}

\begin{proof}
For the sake of contradiction, assume that $f(S) < f(S \cup C) - |C| \cdot \tau$. 
Rewriting this bound in terms of the marginal gain of set $C$ with respect to set $S$ gives:   
$|C| \cdot \tau < f(S \cup C) - f(S) = \Delta(C|S) \le \sum_{v \in C} \Delta(v|S)$, where the second inequality holds because $f$ is a submodular function.
This implies that $\tau < \dfrac{\sum_{v \in C} \Delta(v|S)}{|C|}$, which means that there must be an element $v^* \in C$ whose marginal gain is more than $\tau$. 
Observe that $\Delta(v^*|S) > \tau > 0$ and then $v^* \notin {S}$.
That is, $\Delta(v^*|S) > \tau$ for an element $v^*\in V\setminus S$, 
which contradicts the assumption that we made. Thus, the claim must be true.
\end{proof}

\section{Proof of Lemma~\ref{lem:time}}
\newtheorem*{recall:lem:time}{Lemma~\ref{lem:time}}

\begin{recall:lem:time}
     Let random variable $Q_t$ denote the number of oracle calls that the algorithm \blue{in Theorem~\ref{thm:meta}} makes after the $t$-th update.
    \blue{Then if we choose local search method as \subsetselect{} subroutine we have}
    \begin{align}
    % \[
        \Ex{\sum_{t=1}^T
        Q_t} \le T (\min(k \cdot g(n, k), g(n, k)^2) 
        + \mO(k)) \enspace,
    % \]
        \label{eq:goal_8p}
    \end{align}
     % \red{and for our  $(10+\eps)$-approximation algorithm}
     \blue{and if we choose random sampling method as \subsetselect{} subroutine we have}
    \begin{align}
    % \[
        \Ex{\sum_{t=1}^T
        Q_t} \le T \cdot \min(k \cdot g(n, k), g(n, k)^2) \enspace ,
    % \]
        \label{eq:goal_10p}
    \end{align}
    \blue{where $g(n,k)$ is the number of oracle calls that thresholding algorithm \dynamicthreshold  makes after each update.}
\end{recall:lem:time}

\begin{proof}
    
    We first distinguish the queries depending on the operation that makes the queries.
    We denote the queries made by an update in instance $\mI_1$ as \emph{type-1 queries}, 
    queries made by an update in instance $\mI_2$ as \emph{type-2 queries}, and
    the queries made by \subsetselect{} as \emph{type-3 queries}. 
    We will use $Q_{t}^{(1)}, Q_{t}^{(2)}$ and $Q_{t}^{(3)}$ 
    to denote the number of type-1, type-2 and type-3 queries, the algorithms make 
    for the $t$-th operation. As all of the queries made by the algorithm fall into one of the type-1, type-2, and type-3 categories,
    $Q_t = Q_t^{(1)} + Q_t^{(2)}+ Q_t^{(3)}$. Thus, it 
    suffices to bound $\Ex{\sum_{t=1}^{T} Q_t^{(\ell)}}$ seperately for each $\ell \in \{1, 2, 3\}$.
    
    \textbf{Type-1 queries.}
    We start with $\Ex{\sum_{t=1}^{T} Q_t^{(1)}}$.
    By assumption, the expected amortized number of query calls made by $\mI_1$ is at most $g(n,k)$.
    Thus, $        \Ex{\sum_{t=1}^{T} Q_t^{(1)} }\le T \cdot g(n, k)$.
    
    \textbf{Type-2 queries.}
    We now consider $\sum_{t=1}^{T} Q_t^{(2)}$.
    For each $t$, let random variable $X_t$ denote
    $|S_1 \backslash Z| + |Z \backslash S_1|$ during the
    $t$-th call to \UpdateF{}
    and let
    $v_{t}^{(2), 1}, \dots, v_{t}^{(2), X_t}$ denote
    the elements inserted to or deleted from $\mI_{2}$ during this call.
    In other words, $        \{v_{t}^{(2), 1}, \dots, v_{t}^{(2), X_t}\}
        = 
        \left( S_1 \backslash Z\right)
        \cup
        \left(
         Z \backslash S_1
        \right).$

    For each $1\le i \le X_t$, let $Q_{t}^{(2), i}$ denote the number of oracle calls $\mI_{2}$ makes because of the insertion or deletion of $v_{t}^{(2), i}$. It is clear that
    $Q_t^{(2)} = \sum_{i=1}^{X_t} Q_t^{(2), i}$. 
    By assumption,
    the amortized expected query complexity of $\mI_{2}$ is at most $g(n, k)$. Conditioning on $X_{1}, \dots, X_{T}$, we obtain
    \begin{align*}
        &\Ex{\sum_{t=1}^{T}Q_t^{(2)} | X_{1}, \dots X_{T} }
        = \Ex{(\sum_{t=1}^{T}\sum_{i=1}^{X_t} Q_{t}^{(2), i}) | X_{1}, \dots X_{T}}\\
        & = \sum_{t=1}^{T}\sum_{i=1}^{X_t} \Ex{Q_{t}^{(2), i} | X_{1}, \dots X_{T}}
        \overset{(a)}{\le} \sum_{t=1}^{T}\sum_{i=1}^{X_t} g(n, k) 
        =  (\sum_{t=1}^{T}X_t) \cdot g(n, k) \enspace ,
    \end{align*}
    where for $(a)$, we have used the guarantee on the amortized query complexity of $\mI_2$. 
    Note that as the guarantee holds in the oblivious adversarial model, it holds even after conditioning on $X_{1}, \dots, X_{T}$.
    Taking expectation over $X_{1}, \dots, X_{T}$, we obtain 
    $        \Ex{\sum_{t=1}^{T}Q_t^{(2)}}
        =  g(n, k) \cdot \Ex{\sum_{t=1}^T X_t}$. 
    It remains to bound $\Ex{\sum_{t=1}^{T} X_t}$. As
    $X_t = |S_1 \backslash Z| + |Z \backslash S_1|$, 
    we can always guarantee
    $X_t \le 2k$, leading to the bound 
    $        \Ex{\sum_{t=1}^{T} X_t} \le 2k\cdot T$. 
    In addition, by Property~2 in Definition \ref{def:tau:threshold}, 
    we can guarantee $X_t \le Q_t$, which leads to the bound 
    $\Ex{ \sum_{t=1}^{T} X_t } \le \Ex{\sum_{t=1}^T Q_t^{(1)}} \le T \cdot g(n, k)$.
    Thus, we obtain $        \Ex{\sum_{t=1}^{T}Q_t^{(2)}}
        \le  T \cdot\min(k\cdot g(n, k),  g(n, k)^2)$.
    
    \textbf{Type-3 queries.}
    Finally, we bound 
    $\Ex{\sum_{t=1}^{T}Q_t^{(3)}}$. 
    For the $(8+\eps)$-approximation algorithm, the number of type-3 queries is the 
    number of queries that \subsetselect$(S_1)$ makes, which is $\mO(|S_1|)$. 
    Since, the size of $S_1$ is always bounded by $k$, we then have  
    $\Ex{Q_{t}^{(3)}} \le \mO(|S_1|)= \mO(k)$, 
    which further implies
    $\Ex{\sum_{t=1}^{T}Q_{t}^{(3)}} \le T \cdot \mO(k)$.
       For the $(10+\eps)$-approximation algorithm, since \uniformsubset{} make no queries, the number of type-3 queries is always zero, leading to the bound
    $\Ex{\sum_{t=1}^{T}Q_t^{(3)}} = 0$. 
    
    Combining the above the bounds for type-1, type-2,  and type-3 queries, 
    we obtain the bound 
    \eqref{eq:goal_8p} for the
    $(8+\eps)$-approximation algorithm and the bound 
    \eqref{eq:goal_10p} for the 
    $(10+\eps)$-approximation algorithm. 
\end{proof}

\section{$\tau$-thresholding algorithm}

In this section, we show that the algorithms provided in ~\cite{DBLP:conf/nips/Monemizadeh20} and \cite{bani2023dynamicmat}
are indeed $\tau$-thresholding for any given $\tau$, and can be used in our Metatheorem ~\ref{thm:meta}.

The following is the modified version of the monotone submodular maximization algorithm provided by \cite{bani2023dynamicmat}, ready to be used in our reduction. 

\newcommand{\update}{\textsc{Update}}
\newcommand{\init}{\textsc{Init}}
\newcommand{\carupdates}{\textsc{MonotoneCardinalityConstraintSubroutines}}
\newcommand{\levelingconstraint}{\textsc{MonotoneCardinalityConstraintLeveling}}
\newcommand{\binarysearch}{\textsc{BinarySearch}}
\newcommand{\replacementTester}{\textsc{Promote}}
\newcommand{\marginalgain}[2]{f(#2 + #1) - f(#2)}

\begin{algorithm}[H]
  \caption{\levelingconstraint$(k, \tau)$}
  \label{alg:cardinality:offline_constraint}
  \begin{algorithmic}[1]
    \Function{\init}{$V$}
        \State $I_{0} \gets \emptyset$, \quad $ R_{0} \gets \emptyset$, \quad $T \gets 0$
    \EndFunction
    
\rule{15cm}{0.4pt} 
    \Function{\constLevel}{$i$}
        \State Let $P$ be a random permutation of elements of $R_{i}$ and $\ell \gets i$ 
        \For{$e$ in $P$}\label{line:cardinality:iterate_P}
            \If{ \replacementTester$(I_{\ell-1}, e) = True$} \label{line:cardinality:check_e_is_promoting}
                \State $e_{\ell} \gets e$, \quad  $I_{\ell} \gets I_{\ell-1} + e$, \quad $z \gets \ell$, \quad $R_{\ell+1} \gets  \emptyset$, \quad  $\ell \gets \ell + 1$
                \label{line:cardinality:set_I_l}
            \Else
                \State Run binary search to find the lowest $z \in [i, \ell-1]$ such that \replacementTester$(I_z, e) = False$ 
            \EndIf
            \For{$r \gets i+1$ \textbf{to} $z$}  
                \State $R_r \gets R_r + e$. \label{line:cardinality:constlevelmatroid:addR_bs}
            \EndFor
        \EndFor
        \State $T \gets \ell-1$, which is the final value of $\ell$ in the for-loop above subtracted by one
    \EndFunction
    
\rule{15cm}{0.4pt} 
    \Function{\replacementTester}{$I, e$}
          \If{$\marginalgain{e}{I} \ge \tau$  and  $|I| < k$}
             \State \Return True
          \EndIf
          \State \Return False
    \EndFunction
    \end{algorithmic}
\end{algorithm}

%----------------------------------------------------------------------------------------------------------

\begin{algorithm}
    \caption{\carupdates$(k, \tau)$ }
    \begin{algorithmic}[1]
        \Function{\insertv}{$v$}
        \State $R_0 \gets R_0 + v$. 
        \For{$i \gets 1$ \textbf{to} $T+1$}
            \If{\replacementTester$(I_{i-1}, v)$ = False} \label{line:cardinality:insert:break}
                \State \Break
            \EndIf
            \State $R_{i} \gets R_{i} + v$.
            \State Let $p=1$ with probability $\frac{1}{|R_i|}$, and otherwise $p=0$. \label{line:cardinality:p:insert:klogk}
            \If{$p=1$} \label{line:cardinality:insert:if}
                \State {$e_i \gets v$, \quad $I_i \gets I_{i-1} + v$} \label{line:cardinality:insert:setI}
                \State {$R_{i+1} = \{e' \in R_i: \replacementTester{}(I_i, e') = True \}$} \label{line:cardinality:insert:setRi+1}
                \State {$\constLevel{}(i+1)$
                }
                \State \Break
            \EndIf
        \EndFor
    \EndFunction
    
\rule{15cm}{0.4pt} 

    \Function{\deletev}{$v$}
        \State $R_0 \gets R_0 - v$
        \For{ $i \gets 1$ \textbf{to} $T$}
            \If{$v \notin R_i$}
                \State \Break
            \EndIf
            \State $R_i \gets R_i - v$
            \If{$e_i = v$}
            \State $\constLevel(i)$.
            \State \Break
            \EndIf
        \EndFor
    \EndFunction
    
\rule{15cm}{0.4pt} 
    \Function{\textsc{Extract}}{}
        \State \Return $I_{T}$
    \EndFunction
  \end{algorithmic}
\end{algorithm}

This algorithm maintains a leveled data structure consisting of elements
$e_1, \cdots, e_T$, sets $R_0 \supseteq R_1 \supset \cdots \supset R_T \supset R_{T+1}=\emptyset$, and $I_0, I_1, \cdots, I_T$, $I_i = I_{i - 1} + e_i$ for any $i \in [1, T]$, and $T \leq k$. 
For any $i \in [1, T]$, we have $\marginalgain{e_i}{I_{i-1}} \ge \tau$.
Additionally, if $e \notin I_T$, you can infer that either $T = k$ or $\marginalgain{e}{I_T} < \tau$, considering the submodularity of the function $f$ and the fact that $e$ has been filtered out of $R_{j}$ for some $j \in [1, T]$ due to its low marginal gain with respect to $I_j$ or the size of $I_j$. 
Combining the last two facts establishes Property 1 of Definition \ref{def:tau:threshold}.
Regarding Property 2, when the update is an insertion, the number of changes in the output is either $0$ or at most $|R_{i}|$ if $\constLevel{}(i+1)$ gets invoked. Meanwhile, the number of oracle queries is at least $1$ or $ i + |R_i|$ just by considering the oracle queries made before the invocation of $\constLevel{}$.  
Also, when the update is a deletion, the number of changes in the output is either $0$ or at most $|R_i|$ if $\constLevel{}(i)$ gets invoked. However, the number of oracle queries is either $0$ or at least $R_i$ even with discarding any oracle call made during the binary searches. 
Therefore, this algorithm also satisfies Property 2 of Definition \ref{def:tau:threshold}.
The query complexity of this algorithm is $\mO(klogk)$ per update. [See Lemma~91 in \cite{bani2023dynamicmat} where their parallel runs, which we do not use, are ignored.]

The algorithm of \cite{DBLP:conf/nips/Monemizadeh20} also has a similar construction. 
It constructs its final output $G$ in multiple \emph{levels}.
Initializing $G_0 := \emptyset, R_0 = V$ and $i=1$, it repeatedly performs the following actions.
\begin{enumerate}
    \item  Filtering Step: The algorithm shaves off, i.e., filters, those elements $e \in V$ whose marginal gain with respect to $G_i$ is less than $\tau$ and collects the remaining elements in $R_{i+1}$. In particular,
    $R_{i+1} = \{e \in R_i: \Delta(e | G_i) \ge \tau\}$.
    If there are no remaining elements, the algorithm stops and $G_{i}$ is returned as the final output.
    Otherwise,
    $i$ is increased by one and the algorithm proceeds to the next step
    \item 
    Greedy Step: The algorithm chooses a set $S_i$ taken from $R_i$ uniformly at random. It then creates $G_i$ by adding elements of $S_i$ to $G_{i-1}$ one by one, as long as 
    the marginal gain of the element with respect to $G_i$ is at least $\tau$. In particular, the algorithm initializes $G_i := G_{i-1}$ and then iterates through $S$, adding element $e$ to $G_i$ if $\Delta(e | G_i) \ge \tau$.
    ~\\
    If at any point the size of $G_i$ reaches $k$, the algorithm stops, and $G := G_i$ is returned as the final output.
\end{enumerate}
This approach forms a leveled construction characterized by the sets $\emptyset = G_0 \subset G_1 \subset \dots \subset G_{i^*}=G$ where $i^*$ denotes the final value of $i$ (which is $\mO(k\log(n))$, See Corollary 5 in their paper).
In order to handle updates,
whenever an element in $G$ is deleted, 
the construction is restarted from the level where the deletion occurred. As for insertions, whenever an element $e$ is inserted, the algorithm adds it to each level $i$ as long as it is not filtered beforehand, i.e., $\Delta(e|G_{i-1}) \ge \tau$. Once the element is added, the construction may be restarted with some probability.
We refer to \cite{DBLP:conf/nips/Monemizadeh20} for a more detailed analysis along with the pseudocode.
~\\
Given this description, if $|G| = k$, then $f(G) \ge k \cdot \tau$ as all elements added to $G$ increased $f(G)$ by at least $\tau$, which proves Property 1.a.
As for the case of $|G| < k$, for each $e \in V$, since $e$ was filtered out at some point, it satisfies $\Delta(e | G_i) \le \tau$ for some $i \le i^*$. Since $f$ is submodular and $G_i \subseteq G$, this proves Property 1.b. 
Indeed, these properties are formally proved in Lemma 3 of 
\cite{DBLP:conf/nips/Monemizadeh20}.
~\\
As for Property 2, whenever a level is restarted, the number of changed elements is at most the size of the level (i.e.,  $|R_i|$)
since in the worst case, all elements in $R_i$ are added to $G_i$. 
The number of queries the algorithm makes however is at least the size of the level as the filtering step performs $|R_i|$ queries.
The number of changed elements is therefore bounded by the number of queries as desired.
Finally, the guarantee on the expected amortized query complexity is provided in \cite{DBLP:conf/nips/Monemizadeh20} (Lemma~7 and Lemma~8).

\section{Empirical results}

\subsection{Max-Cut.}
In this section, we study the celebrated Max-Cut problem which is a non-monotone submodular maximization function (See~\cite{DBLP:journals/siamcomp/FeigeMV11}).  
In this problem, given a graph $G(V,E)$,  we would like to compute a set of vertices $S\subseteq V$ that maximizes 
the size of cut $C(S,V \setminus S)$. 

We denote the number of edges in cut $C(S, V \setminus S)$ by $f(S) = |C(S, V \setminus S)|$. 
For the input of our algorithm, we use real-world data sets 
that we collect from SNAP Data Collection~\cite{snapnets}:
loc-Gowalla and web-Stanford. 
Similar to the video summarization experiment, we create a sequence of updates  of 
vertices of graph $G(V,E)$ in a sliding window model. 
For comparing the performance of our algorithm, in addition to \textsc{Sample-Streaming}, we run \textsc{Random Algorithm} that selects $k$ random elements from the remaining elements after each update. 

Finally, we run our algorithm and \textsc{Sample-Streaming} algorithm ten times on different values of $k$. In plots (a) and (c), it is shown that our algorithm has fewer calls than \textsc{Sample-Streaming} when $k$ is small.
In plots (b) and (d), we show that in practice, 
our algorithm has a better approximation guarantee than what we proved theoretically.

As for the non-monotone submodular function, we study the celebrated \textbf{Max-Cut} problem. 
In this problem, given a graph $G(V,E)$,  we would like to compute a set of vertices $S\subseteq V$ that maximizes 
the size of cut $C(S,V \setminus S)$ which is the number of edges between $S$ and $V \setminus S$. 
We denote the number of edges in cut $C(S, V \setminus S)$ by $f(S) = |C(S, V \setminus S)|$.
This function is non-monotone submodular (see~\cite{DBLP:journals/siamcomp/FeigeMV11}). 
Our goal is, given a graph $G(V,E)$ and a parameter $k \in \mathbb{N}$, to compute a Max-Cut $C(S,V \setminus S)$
of $G$ for which the number of vertices in $S$ is at most $|S| \le k$.

For the input of our algorithm, we use real-world data sets 
that we choose from SNAP Data Collection~\cite{snapnets}:
(1) loc-Gowalla (Gowalla location based online social network) that has $196,591$ vertices and $950,327$ edges, and 
(2) web-Stanford (Web graph of Stanford.edu) that has $281,903$ vertices and $2,312,497$ edges. 

We create a sequence $\Xi$ of updates of 
vertices of graph $G(V, E)$. Similar to \cite{DBLP:conf/nips/LattanziMNTZ20}, 
we define the sequence as a sliding window model. 
That is, given a window size $W$, a vertex is inserted at a time $t$ and will be alive for a window of size $W$ 
and then we delete that vertex. 
To make the sequence an adversarial sequence, 
we insert vertices in descending order of their degrees, and each vertex remains alive for a window of $W$ vertex insertions and deletions, 
and later the vertex will be deleted from the graph. 

To evaluate the performance of our algorithm, 
we benchmark the total number of query calls and the submodular value of set $S$ 
of our algorithm and the streaming algorithm proposed for non-monotone submodular maximization so-called \textsc{Sample-Streaming} 
proposed in~\cite{DBLP:conf/nips/FeldmanK018}. 
This algorithm works as follows: 
Upon arrival of an element $u$, with probability $(1-q)$, for a parameter $0 < q < 1$, 
we ignore $u$, otherwise (i.e., with probability $q$), we do the following. 
If the size of set $S$ that we maintain is less than  $k$, i.e, 
$|S|<k$ and $\Delta(u|S) > 0$, we add $u$ to $S$.  
However, if $|S| = k$, we select an element $v \in S$ for which $\Delta(v:S)$ is minimum possible, 
where $\Delta(u:S)$ equals to $\Delta(u | S_u)$ where $S_u$ are elements that arrived before $u$ 
in the stream. If $\Delta(u|S) \ge (1+c) \Delta(v:S)$ for a constant $c$, we replace $v$ by $u$; 
otherwise, we do nothing. 
Now we convert this streaming algorithm into a dynamic algorithm. 
For that, we restart Sample-Streaming after every deletion that deletes a vertex of 
reported set $S$ by \textsc{Sample-Streaming}'s outputs. That is, if a deletion does not 
touch any vertex in set $S$, we do nothing; otherwise we restart the streaming algorithm.

Finally, we run our algorithm for $\eps=1$ and the \textsc{Sample-Streaming} algorithm ten times on different values of $k$ 
and benchmark them in Figure~\ref{fig:experiment}.
In plots (a) and (c), it is shown that our algorithm has fewer oracle calls than \textsc{Sample-Streaming} when $k$ is small.
The number of oracle calls of the \textsc{Sample-Streaming} algorithm is $\mO(n^2)$ in the worst case which is independent of $k$.
Therefore, by increasing $k$, 
the number of oracle calls of our algorithm becomes worse than the number of oracle calls of the Sample-Streaming algorithm. 
On the other hand, our algorithm works significantly better when the size of the input graph increases.
As is shown in Figure~\ref{fig:experiment}, in plots (a) and (c), 
running on a bigger graph increases the number of oracle calls of \textsc{Sample-Streaming} 
more than our algorithm.

\begin{figure}
\centering
%\subfigure[loc-Gowalla's total oracle calls]{
\begin{tabular}{@{}c@{}}
\begin{tikzpicture}
\pgfplotsset{width=6cm,compat=1.9}
\begin{axis}[
    xlabel={$k$},
    ylabel={oracle calls},
    ymin=0, ymax=5300000000,
     xtick={10, 20, 30, 40, 50},
    legend pos=north west,
    ymajorgrids=true,
    grid style=dashed,
    legend style={nodes={scale=0.5, transform shape}},
]
% \addplot[
%     color=green,
%     mark=triangle,
%     error bars/.cd,
%     y dir=both,
%     y explicit,
%     error bar style={
%     color=green}
%     ]
%     coordinates {
%         ( 5 , 9064275.500000 ) += (0, 641804.500000) -= (0, 641804.500000)
% 		( 10 , 63649028.000000 ) += (0, 116381.000000) -= (0, 116381.000000)
% 		( 15 , 169927895.000000 ) += (0, 8721723.000000) -= (0, 8721723.000000)
% 		( 20 , 333834909.000000 ) += (0, 2164947.000000) -= (0, 2164947.000000)
% 		( 25 , 538296378.500000 ) += (0, 5720781.500000) -= (0, 5720781.500000)
% 		( 30 , 920026970.000000 ) += (0, 1038850.000000) -= (0, 1038850.000000)
% 		( 35 , 1317759302.500000 ) += (0, 12886488.500000) -= (0, 12886488.500000)
% 		( 40 , 1818488987.000000 ) += (0, 0.000000) -= (0, 0.000000)
% 		( 45 , 2487083690.000000 ) += (0, 0.000000) -= (0, 0.000000)
%     };
%     \addlegendentry{Our Dynamic Algorithm,
%     $\epsilon=2$}

\addplot[
    color=blue,
    mark=star,
    error bars/.cd,
    y dir=both,
    y explicit,
    error bar style={
    color=blue}
    ]
    coordinates {
        ( 5 , 26217507.700000 ) += (0, 1659405.300000) -= (0, 1471015.700000)
		( 10 , 145281586.200000 ) += (0, 8861579.800000) -= (0, 5122754.200000)
		( 15 , 359184099.300000 ) += (0, 8080608.700000) -= (0, 6672706.300000)
		( 20 , 736363845.800000 ) += (0, 28792657.200000) -= (0, 55878251.800000)
		( 25 , 1213196337.200000 ) += (0, 63481979.800000) -= (0, 55120812.200000)
		( 30 , 1814655289.600000 ) += (0, 67505990.400000) -= (0, 64998651.600000)
		( 35 , 2741748775.400000 ) += (0, 73827638.600000) -= (0, 64166745.400000)
		( 40 , 3751847649.100000 ) += (0, 131622916.900000) -= (0, 147327279.100000)
		( 45 , 4979103069.700000 ) += (0, 160551026.300000) -= (0, 127589011.700000)
    };
    \addlegendentry{Our Dynamic Algorithm,
    $\epsilon=1$}

\addplot[
    color=red,
    mark=square,
    error bars/.cd,
    y dir=both,
    y explicit,
    error bar style={
    color=red}
    ]
    coordinates {
        ( 5 , 1560997415.900000 ) += (0, 8648623.100000) -= (0, 7669959.900000)
		( 10 , 1567314201.300000 ) += (0, 7456103.700000) -= (0, 6349165.300000)
		( 15 , 1564426698.900000 ) += (0, 14870200.100000) -= (0, 9670586.900000)
		( 20 , 1566518139.700000 ) += (0, 12721529.300000) -= (0, 12004174.700000)
		( 25 , 1564419391.600000 ) += (0, 11066482.400000) -= (0, 9514318.600000)
		( 30 , 1564943395.400000 ) += (0, 8977314.600000) -= (0, 7967891.400000)
		( 35 , 1564653082.900000 ) += (0, 14254919.100000) -= (0, 7207053.900000)
		( 40 , 1564550144.900000 ) += (0, 11065720.100000) -= (0, 10902292.900000)
		( 45 , 1567986201.100000 ) += (0, 8378287.900000) -= (0, 6284425.100000) 
    };
    \addlegendentry{\textsc{Sample-Streaming}}

\addplot[
    color=yellow,
    mark=triangle,
    ]
    coordinates {
        (5, 393183) (10, 393183) (15, 393183) (20, 393183) (25, 393183) (30, 393183) (35, 393183) (40, 393183) (45, 393183) 
    };
    \addlegendentry{\textsc{Random Algorithm}}

\end{axis}
\end{tikzpicture}%}%
% \caption{loc-Gowalla's total oracle calls}
\end{tabular}

\small (a) loc-Gowalla's total oracle calls

% 
%\subfigure[loc-Gowalla's average output]{
\begin{tabular}{@{}c@{}}
\begin{tikzpicture}
\pgfplotsset{width=6cm,compat=1.9}
\begin{axis}[
    xlabel={$k$},
    ylabel={f},
    ymin=0, ymax=100000,
    xtick={10, 20, 30, 40, 50},
    legend pos=north west,
    ymajorgrids=true,
    grid style=dashed,
    legend style={nodes={scale=0.5, transform shape}},
]

% \addplot[
%     color=green,
%     mark=triangle,
%     error bars/.cd,
%     y dir=both,
%     y explicit,
%     error bar style={
%     color=green}
%     ]
%     coordinates {
%         ( 10 , 19140.300000 ) += (0, 0.600000) -= (0, 0.600000)
% 		( 15 , 27020.250000 ) += (0, 640.750000) -= (0, 640.750000)
% 		( 20 , 32300.450000 ) += (0, 2.050000) -= (0, 2.050000)
% 		( 25 , 34587.300000 ) += (0, 2.000000) -= (0, 2.000000)
% 		( 30 , 36239.000000 ) += (0, 0.200000) -= (0, 0.200000)
% 		( 35 , 41851.700000 ) += (0, 1.400000) -= (0, 1.400000)
% 		( 40 , 46541.300000 ) += (0, 0.000000) -= (0, 0.000000)
% 		( 45 , 50838.000000 ) += (0, 0.000000) -= (0, 0.000000)
%     };
%     \addlegendentry{Our Dynamic Algorithm,
%     $\epsilon=2$}

\addplot[
    color=blue,
    mark=star,
    error bars/.cd,
    y dir=both,
    y explicit,
    error bar style={
    color=blue}
    ]
    coordinates {
        ( 5 , 26978.040000 ) += (0, 2236.560000) -= (0, 2633.540000)
		( 10 , 35520.800000 ) += (0, 2949.300000) -= (0, 3336.400000)
		( 15 , 37055.500000 ) += (0, 5334.200000) -= (0, 2578.600000)
		( 20 , 41728.740000 ) += (0, 372.160000) -= (0, 42.040000)
		( 25 , 50624.030000 ) += (0, 0.870000) -= (0, 0.830000)
		( 30 , 54330.310000 ) += (0, 1.590000) -= (0, 1.510000)
		( 35 , 61221.630000 ) += (0, 1.570000) -= (0, 0.730000)
		( 40 , 63812.220000 ) += (0, 1105.980000) -= (0, 5189.720000)
		( 45 , 60465.320000 ) += (0, 5156.780000) -= (0, 4600.920000) 
    };
    \addlegendentry{Our Dynamic Algorithm,
    $\epsilon=1$}

\addplot[
    color=red,
    mark=square,
    error bars/.cd,
    y dir=both,
    y explicit,
    error bar style={
    color=red}
    ]
    coordinates {
        ( 5 , 13638.679000 ) += (0, 6572.721000) -= (0, 6392.689000)
		( 10 , 19149.889000 ) += (0, 12803.211000) -= (0, 9518.199000)
		( 15 , 24159.830000 ) += (0, 11187.970000) -= (0, 10216.130000)
		( 20 , 25992.370000 ) += (0, 14218.230000) -= (0, 8841.070000)
		( 25 , 24390.730000 ) += (0, 18899.370000) -= (0, 6779.430000)
		( 30 , 29855.560000 ) += (0, 10300.040000) -= (0, 10453.360000)
		( 35 , 32096.300000 ) += (0, 9838.600000) -= (0, 8658.100000)
		( 40 , 30935.600000 ) += (0, 6829.600000) -= (0, 8050.600000)
		( 45 , 33399.360000 ) += (0, 12104.240000) -= (0, 7742.260000) 
    };
    \addlegendentry{\textsc{Sample-Streaming}}

\addplot[
    color=yellow,
    mark=triangle,
    ]
    coordinates {
        (5, 163.129) (10, 323.8) (15, 484.162) (20, 644.309) (25, 805.141) (30, 964.978) (35, 1123.22) (40, 1285.19) (45, 1437.75) 
    };
    \addlegendentry{\textsc{Random Algorithm}}
    
\end{axis}
\end{tikzpicture}%}
% \caption*{loc-Gowalla's average output}
\end{tabular}

\small (b) loc-Gowalla's average output

%\subfigure[web-Stanford's total oracle calls]{
\begin{tabular}{@{}c@{}}
\begin{tikzpicture}
\pgfplotsset{width=6cm,compat=1.9}
\begin{axis}[
    xlabel={$k$},
    ylabel={oracle calls},
    ymin=0, ymax=7300000000,
     xtick={10, 20, 30, 40, 50},
    legend pos=north west,
    ymajorgrids=true,
    grid style=dashed,
    legend style={nodes={scale=0.5, transform shape}},
]
\addplot[
    color=green,
    mark=triangle,
    error bars/.cd,
    y dir=both,
    y explicit,
    error bar style={
    color=green}
    ]
    coordinates {
        ( 5 , 13189332.666667 ) += (0, 600024.333333) -= (0, 1077809.666667)
		( 10 , 92910830.555556 ) += (0, 11086284.444444) -= (0, 9218599.555556)
		( 15 , 253285739.111111 ) += (0, 16120940.888889) -= (0, 9157152.111111)
		( 20 , 494871649.000000 ) += (0, 39629962.000000) -= (0, 32443099.000000)
		( 25 , 869112984.888889 ) += (0, 37816876.111111) -= (0, 53043565.888889)
		( 30 , 1348469370.666667 ) += (0, 86707924.333333) -= (0, 94798342.666667)
		( 35 , 1988308667.000000 ) += (0, 85599666.000000) -= (0, 78343418.000000)
		( 40 , 2629199650.555555 ) += (0, 104513000.444445) -= (0, 53913724.555555)
		( 45 , 3517993545.666667 ) += (0, 139320140.333333) -= (0, 188208389.666667)
    };
    \addlegendentry{Our Dynamic Algorithm,
    $\epsilon=2$}
    
\addplot[
    color=blue,
    mark=star,
    error bars/.cd,
    y dir=both,
    y explicit,
    error bar style={
    color=blue}
    ]
    coordinates {
        ( 5 , 35405554.222222 ) += (0, 2376443.777778) -= (0, 1486295.222222)
		( 10 , 196508420.222222 ) += (0, 7381699.777778) -= (0, 10420606.222222)
		( 15 , 515682379.333333 ) += (0, 37699153.666667) -= (0, 33594104.333333)
		( 20 , 1033854317.555556 ) += (0, 56954701.444444) -= (0, 38228958.555556)
		( 25 , 1702475286.000000 ) += (0, 54973038.000000) -= (0, 44360331.000000)
		( 30 , 2611829259.000000 ) += (0, 85458393.000000) -= (0, 44244028.000000)
		( 35 , 3742451481.777778 ) += (0, 170672852.222222) -= (0, 171814586.777778)
		( 40 , 5174711541.666667 ) += (0, 74436428.333333) -= (0, 97808016.666667)
		( 45 , 6807616127.777778 ) += (0, 198154404.222222) -= (0, 218206573.777778)
    };
    \addlegendentry{Our Dynamic Algorithm,
    $\epsilon=1$}

\addplot[
    color=red,
    mark=square,
    error bars/.cd,
    y dir=both,
    y explicit,
    error bar style={
    color=red}
    ]
    coordinates {
        ( 5 , 3114011915.444445 ) += (0, 17899720.555555) -= (0, 20116425.444445)
		( 10 , 3117948866.888889 ) += (0, 14555063.111111) -= (0, 21233219.888889)
		( 15 , 3126764897.666667 ) += (0, 23146108.333333) -= (0, 22656578.666667)
		( 20 , 3118338550.888889 ) += (0, 9580419.111111) -= (0, 13787916.888889)
		( 25 , 3124633826.111111 ) += (0, 13504537.888889) -= (0, 12850275.111111)
		( 30 , 3125302115.333333 ) += (0, 8645101.666667) -= (0, 9255341.333333)
		( 35 , 3124813521.888889 ) += (0, 5886252.111111) -= (0, 12436652.888889)
		( 40 , 3126550079.777778 ) += (0, 12202176.222222) -= (0, 15273603.777778)
		( 45 , 3125235661.222222 ) += (0, 12601548.777778) -= (0, 11039934.222222)
    };
    \addlegendentry{\textsc{Sample-Streaming}}

\addplot[
    color=yellow,
    mark=triangle,
    ]
    coordinates {
        (5, 393182) (10, 393183) (15, 393183) (20, 393183) (25, 393183) (30, 393183) (35, 393183) (40, 393183) (45, 393183)
    };
    \addlegendentry{\textsc{Random Algorithm}}
    
\end{axis}
\end{tikzpicture}%}%
% \caption*{web-Stanford's total oracle calls}
\end{tabular}

\small (c) web-Stanford's total oracle calls

%\subfigure[web-Stanford's average output]{
\begin{tabular}{@{}c@{}}
\begin{tikzpicture}
\pgfplotsset{width=6cm,compat=1.9}
\begin{axis}[
    xlabel={$k$},
    ylabel={f},
    ymin=0, ymax=250000,
    xtick={10, 20, 30, 40, 50},
    legend pos=north west,
    ymajorgrids=true,
    grid style=dashed,
    legend style={nodes={scale=0.5, transform shape}},
]
\addplot[
    color=green,
    mark=triangle,
    error bars/.cd,
    y dir=both,
    y explicit,
    error bar style={
    color=green}
    ]
    coordinates {
        ( 5 , 24377.088889 ) += (0, 8212.411111) -= (0, 3746.788889)
		( 10 , 70443.922222 ) += (0, 4070.877778) -= (0, 6677.822222)
		( 15 , 78713.766667 ) += (0, 6524.833333) -= (0, 8176.566667)
		( 20 , 98380.500000 ) += (0, 7992.500000) -= (0, 11087.700000)
		( 25 , 106610.011111 ) += (0, 4841.988889) -= (0, 9088.211111)
		( 30 , 133892.888889 ) += (0, 6230.111111) -= (0, 5117.888889)
		( 35 , 148297.444444 ) += (0, 2814.555556) -= (0, 2955.444444)
		( 40 , 139530.777778 ) += (0, 10998.222222) -= (0, 9124.777778)
		( 45 , 145499.777778 ) += (0, 4852.222222) -= (0, 4439.777778)
    };
    \addlegendentry{Our Dynamic Algorithm, $\epsilon = 2$}
    
\addplot[
    color=blue,
    mark=star,
    error bars/.cd,
    y dir=both,
    y explicit,
    error bar style={
    color=blue}
    ]
    coordinates {
        ( 5 , 37587.422222 ) += (0, 3103.477778) -= (0, 3924.622222)
		( 10 , 67678.411111 ) += (0, 5002.088889) -= (0, 4951.511111)
		( 15 , 99352.288889 ) += (0, 2204.711111) -= (0, 906.088889)
		( 20 , 98558.366667 ) += (0, 4564.633333) -= (0, 613.366667)
		( 25 , 123995.666667 ) += (0, 6879.333333) -= (0, 5083.666667)
		( 30 , 131498.555556 ) += (0, 6824.444444) -= (0, 9953.555556)
		( 35 , 154508.555556 ) += (0, 1.444444) -= (0, 1.555556)
		( 40 , 157116.888889 ) += (0, 1.111111) -= (0, 0.888889)
		( 45 , 156808.222222 ) += (0, 6839.777778) -= (0, 4863.222222)
    };
    \addlegendentry{Our Dynamic Algorithm, $\epsilon = 1$}

\addplot[
    color=red,
    mark=square,
    error bars/.cd,
    y dir=both,
    y explicit,
    error bar style={
    color=red}
    ]
    coordinates {
        ( 5 , 31769.544444 ) += (0, 8921.555556) -= (0, 9198.744444)
		( 10 , 42163.533333 ) += (0, 17786.866667) -= (0, 16716.933333)
		( 15 , 49351.222222 ) += (0, 21029.377778) -= (0, 15013.822222)
		( 20 , 59337.433333 ) += (0, 20001.266667) -= (0, 14852.433333)
		( 25 , 53153.088889 ) += (0, 16197.111111) -= (0, 21085.388889)
		( 30 , 63739.488889 ) += (0, 20565.411111) -= (0, 22483.688889)
		( 35 , 61544.600000 ) += (0, 18584.100000) -= (0, 18290.500000)
		( 40 , 66791.822222 ) += (0, 16393.277778) -= (0, 29546.522222)
		( 45 , 73836.966667 ) += (0, 21812.133333) -= (0, 13626.366667)
    };
    \addlegendentry{\textsc{Sample-Streaming}}

\addplot[
    color=yellow,
    mark=triangle,
    ]
    coordinates {
        (5, 168.531) (10, 337.935) (15, 502.889) (20, 670.701) (25, 836.278) (30, 1002.38) (35, 1162.24) (40, 1326.02) (45, 1489.47) 
    };
    \addlegendentry{\textsc{Random Algorithm}}
    
\end{axis}
\end{tikzpicture}%}
% \caption*{web-Stanford's average output}
\end{tabular}

\small (d) web-Stanford's average output 

\caption{We plots the total number of query calls and the average output of our dynamic algorithm, \textsc{Sample-Streaming}, and \textsc{Random Algorithm} on two data sets.} 
\label{fig:experiment}
\end{figure}

%%%%%%%

In addition, in plots (b) and (d), we show that in practice, 
our algorithm has a better approximation guarantee than what we proved theoretically. 
Interestingly, the approximation factor of \textsc{Sample-Streaming} is $3+2\sqrt{2} \approx 5.828$ 
which is better than the approximation factor that we proved for our dynamic algorithm, but 
the plots (b) and (d) show that in reality, our algorithm performs better.
Moreover, we run \textsc{Random Algorithm} that selects $k$ random elements from remaining elements after each update. 
As we see in Figure~\ref{fig:experiment} plots (b) and (d), \textsc{Random algorithm} performs poorly. 
Also, we run our algorithm for $\eps=2$. 
As is shown in plot (d), the output has not changed significantly by increasing $\eps$, and it remains more than the average output of \textsc{Sample-Streaming}.
On the other hand, the number of oracle calls decreased significantly (plot (c)).
Therefore, for bigger $k$, increasing the $\eps$ is a good option to get an appropriate output in a short time.

\newpage

\subsection{Additional experiments}
In this section, in addition to our previous experiments, we run another experiment.
In Figure \ref{fig:sup_experiment}, we use a slightly different type of input 
where we first insert all vertices in descending order of their degrees. 
Then, we delete them in a semi-randomized order.
For deletions, we first sort vertices in descending order of their degrees (ties are broken arbitrarily) and with half probability 
we swap a vertex with its left or right neighbors. 
This will be the order that we follow to delete vertices. 
We use this type of input to make our input more general, such that no sliding windows algorithm works on it. 
Also, the noise that we added to the deletion makes our input more real.

\begin{figure}[h]
\centering

\begin{tabular}{@{}c@{}}
\begin{tikzpicture}
\pgfplotsset{width=6cm,compat=1.9}
 \begin{axis}
[
     xlabel={$k$},
     ylabel={oracle calls},
     ymin=0, ymax=5300000000,
     xtick={10, 20, 30, 40, 50},
     legend pos=north west,
     ymajorgrids=true,
     grid style=dashed,
     legend style={nodes={scale=0.5, transform shape}}
]

\addplot[
    color=red,
    mark=square,
    ]
    coordinates {
        (5, 484371804) (10, 1462307127) (15, 1475462952) (20, 1469936061) (25, 1477983557) (30, 1479531073) (35, 1473047669) (40, 1464420175) (45, 1478238401) 
    };
    \addlegendentry{Sample-Streaming}

\addplot[
    color=blue,
    mark=star,
    ]
    coordinates {
        (5, 24598208) (10, 146169121) (15, 339349878) (20, 715005346) (25, 1256514103) (30, 1907653364) (35, 2811814891) (40, 3644061084) (45, 4935024520) 
    };
    \addlegendentry{\textsc{Our Dynamic Algorithm}}

\addplot[
    color=yellow,
    mark=triangle,
    ]
    coordinates {
        (5, 393182) (10, 393183) (15, 393183) (20, 393183) (25, 393183) (30, 393183) (35, 393183) (40, 393183) (45, 393183)
    };
    \addlegendentry{\textsc{Random Algorithm}}
    
\end{axis}
\end{tikzpicture}

\end{tabular}

\small (a) random loc-Gowalla's total oracle calls

\begin{tabular}{@{}c@{}}
\begin{tikzpicture}
\pgfplotsset{width=6cm,compat=1.9}
\begin{axis}[
    xlabel={$k$},
    ylabel={f},
    ymin=0, ymax=100000,
    xtick={10, 20, 30, 40, 50},
    legend pos=north west,
    ymajorgrids=true,
    grid style=dashed,
    legend style={nodes={scale=0.5, transform shape}},
]

\addplot[
    color=red,
    mark=square,
    ]
    coordinates {
        (5, 12263.4) (10, 13334.3) (15, 24169.7) (20, 19319.8) (25, 31192.5) (30, 32695.8) (35, 53792.9) (40, 51480.6) (45, 36781.7) 
    };
    \addlegendentry{Sample-Streaming}

\addplot[
    color=blue,
    mark=star,
    ]
    coordinates {
        (5, 38052.3) (10, 47144.6) (15, 54282.6) (20, 54427.3) (25, 66086.6) (30, 70893.3) (35, 79876.7) (40, 84664.9) (45, 75983.8)  
    };
    \addlegendentry{\textsc{Our Dynamic Algorithm}}

\addplot[
    color=yellow,
    mark=triangle,
    ]
    coordinates {
        (5, 168.531) (10, 337.935) (15, 502.889) (20, 670.701) (25, 836.278) (30, 1002.38) (35, 1162.24) (40, 1326.02) (45, 1489.47) 
    };
    \addlegendentry{\textsc{Random Algorithm}}
    
\end{axis}
\end{tikzpicture}
\end{tabular}

\small (b) random loc-Gowalla's average output

\caption{We plot the total number of query calls and the average output of our dynamic algorithm, \textsc{Sample-Streaming}, and \textsc{Random Algorithm} on loc-Gowalla with a noisy order for deletion.}
\label{fig:sup_experiment}
\end{figure}

\newpage

\subsection{Larger version of video summarization plots}
Here in Figure \ref{fig:video_experiment_benchmark_large} we provided a larger version of Figure \ref{fig:video_experiment_benchmark} to make it easier to read.

\begin{figure*}[h]
\centering
\begin{tabular}{@{}c@{}}
\begin{tikzpicture}
\pgfplotsset{width=8cm,compat=1.9}
\begin{axis}[
    xlabel={$k$},
    ylabel={oracle calls},
    ymin=0, ymax=1000,
    xtick={2, 4, 6, 8, 10},
    legend pos=north west,
    ymajorgrids=true,
    grid style=dashed,
    legend style={nodes={scale=0.5, transform shape}},
]

\addplot[
    color=red,
    mark=square,
    error bars/.cd,
    y dir=both,
    y explicit,
    error bar style={
    color=red}
    ]
    coordinates {
        (1, 331.20000) += (0, 106.80000) -= (0, 123.20000) (2, 373.60000) += (0, 108.40000) -= (0, 143.60000) (3, 346.80000) += (0, 142.20000) -= (0, 64.80000) (4, 371.00000) += (0, 77.00000) -= (0, 124.00000) (5, 329.80000) += (0, 94.20000) -= (0, 89.80000) (6, 309.50000) += (0, 50.50000) -= (0, 50.50000) (7, 368.80000) += (0, 80.20000) -= (0, 127.80000) (8, 352.70000) += (0, 54.30000) -= (0, 74.70000) (9, 367.60000) += (0, 112.40000) -= (0, 104.60000) (10, 385.60000) += (0, 169.40000) -= (0, 94.60000) 
    };
    \addlegendentry{\textsc{Sample-Streaming}}

\addplot[
    color=blue,
    mark=star,
    error bars/.cd,
    y dir=both,
    y explicit,
    error bar style={
    color=blue}
    ]
    coordinates {
        (1, 461.40000) += (0, 9.60000) -= (0, 15.40000) (2, 662.30000) += (0, 90.70000) -= (0, 99.30000) (3, 529.50000) += (0, 91.50000) -= (0, 64.50000) (4, 447.80000) += (0, 31.20000) -= (0, 54.80000) (5, 323.90000) += (0, 20.10000) -= (0, 24.90000) (6, 235.40000) += (0, 16.60000) -= (0, 12.40000) (7, 179.20000) += (0, 9.80000) -= (0, 16.20000) (8, 178.50000) += (0, 18.50000) -= (0, 10.50000) (9, 189.80000) += (0, 16.20000) -= (0, 16.80000) (10, 238.30000) += (0, 14.70000) -= (0, 18.30000)
    };
    \addlegendentry{\textsc{Our Dynamic Algorithm}}
\end{axis}
\end{tikzpicture}

\begin{tikzpicture}
\pgfplotsset{width=8cm,compat=1.9}
\begin{axis}[
    xlabel={$k$},
    ylabel={f},
    ymin=0, ymax=8,
    xtick={2, 4, 6, 8, 10},
    legend pos=north west,
    ymajorgrids=true,
    grid style=dashed,
    legend style={nodes={scale=0.5, transform shape}},
]

\addplot[
    color=red,
    mark=square,
    error bars/.cd,
    y dir=both,
    y explicit,
    error bar style={
    color=red}
    ]
    coordinates {
        (1, 1.62818) += (0, 0.44989) -= (0, 0.57743) (2, 2.14678) += (0, 0.65564) -= (0, 0.41891) (3, 2.48939) += (0, 0.33780) -= (0, 0.40127) (4, 2.39747) += (0, 0.34600) -= (0, 0.64868) (5, 2.28432) += (0, 0.52396) -= (0, 0.47845) (6, 2.30636) += (0, 0.45690) -= (0, 0.35722) (7, 2.43273) += (0, 0.48661) -= (0, 0.59661) (8, 2.39775) += (0, 0.46000) -= (0, 0.46214) (9, 2.28745) += (0, 0.42656) -= (0, 0.32351) (10, 2.39145) += (0, 0.50267) -= (0, 0.74555) 
    };
    \addlegendentry{\textsc{Sample-Streaming}}

\addplot[
    color=blue,
    mark=star,
    error bars/.cd,
    y dir=both,
    y explicit,
    error bar style={
    color=blue}
    ]
    coordinates {
        (1, 2.46816) += (0, 0.07933) -= (0, 0.42347) (2, 2.15307) += (0, 0.32738) -= (0, 0.17542) (3, 2.29888) += (0, 0.09777) -= (0, 0.19832) (4, 2.19944) += (0, 0.43743) -= (0, 0.58491) (5, 2.39385) += (0, 0.40503) -= (0, 0.38268) (6, 2.36480) += (0, 0.17151) -= (0, 0.35363) (7, 2.47486) += (0, 0.06145) -= (0, 0.17877) (8, 2.49553) += (0, 0.03520) -= (0, 0.02626) (9, 2.46536) += (0, 0.07095) -= (0, 0.25866) (10, 2.26592) += (0, 0.20335) -= (0, 0.26033)
    };
    \addlegendentry{\textsc{Our Dynamic Algorithm}}
\end{axis}
\end{tikzpicture}
\end{tabular}

\small (a) Video 106 total oracle calls and average output

\begin{tabular}{@{}c@{}}
\begin{tikzpicture}
\pgfplotsset{width=8cm,compat=1.9}
\begin{axis}[
    xlabel={$k$},
    ylabel={oracle calls},
    ymin=0, ymax=1300,
     xtick={2, 4, 6, 8, 10},
    legend pos=north west,
    ymajorgrids=true,
    grid style=dashed,
    legend style={nodes={scale=0.5, transform shape}},
]

\addplot[
    color=red,
    mark=square,
    error bars/.cd,
    y dir=both,
    y explicit,
    error bar style={
    color=red}
    ]
    coordinates {
        (1, 322.70000) += (0, 146.30000) -= (0, 146.70000) (2, 484.00000) += (0, 183.00000) -= (0, 309.00000) (3, 553.50000) += (0, 217.50000) -= (0, 183.50000) (4, 687.60000) += (0, 171.40000) -= (0, 182.60000) (5, 603.70000) += (0, 118.30000) -= (0, 142.70000) (6, 625.30000) += (0, 94.70000) -= (0, 112.30000) (7, 681.60000) += (0, 220.40000) -= (0, 104.60000) (8, 615.80000) += (0, 133.20000) -= (0, 196.80000) (9, 717.10000) += (0, 173.90000) -= (0, 180.10000) (10, 694.10000) += (0, 170.90000) -= (0, 180.10000) 
    };
    \addlegendentry{\textsc{Sample-Streaming}}

\addplot[
    color=blue,
    mark=star,
    error bars/.cd,
    y dir=both,
    y explicit,
    error bar style={
    color=blue}
    ]
    coordinates {
        (1, 445.80000) += (0, 14.20000) -= (0, 15.80000) (2, 654.50000) += (0, 86.50000) -= (0, 56.50000) (3, 580.50000) += (0, 41.50000) -= (0, 51.50000) (4, 574.30000) += (0, 60.70000) -= (0, 51.30000) (5, 675.40000) += (0, 156.60000) -= (0, 88.40000) (6, 644.90000) += (0, 89.10000) -= (0, 79.90000) (7, 581.80000) += (0, 52.20000) -= (0, 39.80000) (8, 520.90000) += (0, 21.10000) -= (0, 47.90000) (9, 583.10000) += (0, 53.90000) -= (0, 55.10000) (10, 706.50000) += (0, 173.50000) -= (0, 107.50000) 
    };
    \addlegendentry{\textsc{Our Dynamic Algorithm}}
\end{axis}
\end{tikzpicture}

\begin{tikzpicture}
\pgfplotsset{width=8cm,compat=1.9}
\begin{axis}[
    xlabel={$k$},
    ylabel={f},
    ymin=0, ymax=10,
    xtick={2, 4, 6, 8, 10},
    legend pos=north west,
    ymajorgrids=true,
    grid style=dashed,
    legend style={nodes={scale=0.5, transform shape}},
]

\addplot[
    color=red,
    mark=square,
    error bars/.cd,
    y dir=both,
    y explicit,
    error bar style={
    color=red}
    ]
    coordinates {
        (1, 2.20679) += (0, 0.38838) -= (0, 0.38206) (2, 3.30840) += (0, 0.44028) -= (0, 0.39134) (3, 3.51886) += (0, 0.73748) -= (0, 0.68497) (4, 3.22474) += (0, 0.36300) -= (0, 0.61065) (5, 2.83012) += (0, 0.86126) -= (0, 1.18679) (6, 2.91402) += (0, 0.50340) -= (0, 1.78629) (7, 2.96153) += (0, 0.91239) -= (0, 1.41191) (8, 3.19881) += (0, 0.46492) -= (0, 0.93366) (9, 3.18183) += (0, 0.96342) -= (0, 1.56121) (10, 3.23328) += (0, 0.66669) -= (0, 0.46605)
    };
    \addlegendentry{\textsc{Sample-Streaming}}

\addplot[
    color=blue,
    mark=star,
    error bars/.cd,
    y dir=both,
    y explicit,
    error bar style={
    color=blue}
    ]
    coordinates {
        (1, 3.34082) += (0, 0.00204) -= (0, 0.00613) (2, 3.62122) += (0, 0.31347) -= (0, 0.27428) (3, 4.40612) += (0, 0.52449) -= (0, 0.53673) (4, 3.27633) += (0, 0.29102) -= (0, 0.32123) (5, 3.82531) += (0, 0.28081) -= (0, 0.55184) (6, 4.18082) += (0, 0.52530) -= (0, 0.44204) (7, 4.47143) += (0, 0.29184) -= (0, 0.36531) (8, 4.73959) += (0, 0.41143) -= (0, 0.69061) (9, 4.56939) += (0, 0.92449) -= (0, 0.63878) (10, 4.52245) += (0, 0.91020) -= (0, 0.66531)
    };
    \addlegendentry{\textsc{Our Dynamic Algorithm}}
\end{axis}
\end{tikzpicture}
\end{tabular}

\small (b) Video 36 total oracle calls and average output

\caption{We plot the total number of query calls and the average output of our dynamic algorithm and \textsc{Sample-Streaming} 
on video 106 from YouTube and video 36 from OVP.}
\label{fig:video_experiment_benchmark_large}
\end{figure*}

\end{document}